\documentclass[11pt]{article}
\usepackage{mathrsfs}
\usepackage{amsmath}
\usepackage[colorlinks=true, allcolors=blue, urlcolor=black]{hyperref}

\makeatletter
\newcommand{\removelatexerror}{\let\@latex@error\@gobble}
\makeatother
\usepackage{amsmath}
\usepackage{amsthm}
\usepackage{amsfonts}
\usepackage{amssymb}
\usepackage{graphicx}
\usepackage{dsfont}
\usepackage{mathtools}
\usepackage{cleveref}
\usepackage{pgfplots}
\usepackage{bbm}
\usepackage[noend]{algpseudocode}
\usepackage{stmaryrd}

\usepackage{algorithm}
\usepackage{complexity}
\usepackage{enumitem}
\usepackage{tikz}
\usetikzlibrary{positioning}
\usepackage{thm-restate}
\usepackage{float}
\usetikzlibrary{calc}
\usepackage{thm-restate}
\usetikzlibrary{shapes,decorations.markings}
\usepackage{caption}
\usepackage{subcaption}
\usepackage{comment}

\usepackage{lipsum}
\usetikzlibrary{snakes}

\newtheorem{theorem}{Theorem}[section]

\newtheorem{definition}[theorem]{Definition}
\newtheorem{claim}[theorem]{Claim}
\newtheorem{lemma}[theorem]{Lemma}
\newtheorem{proposition}[theorem]{Proposition}
\newtheorem{question}[theorem]{Question}

\newtheorem{fact}[theorem]{Fact}
\newtheorem{corollary}[theorem]{Corollary}

\renewcommand{\E}{\mathbb{E}}
\newcommand{\Var}{\textnormal{Var}}
\newcommand{\Varsub}[1]{\underset{#1}{\Var}}
\newcommand{\DISJ}{\mathsf{DISJ}}
\newcommand{\IND}{\mathsf{IND}}
\newcommand{\Esub}[1]{\underset{#1}{\E}}
\newcommand{\GapMAJ}{\mathsf{GapMAJ}}
\newcommand{\GH}{\mathsf{GH}}

\renewcommand{\A}{\mathcal{A}}
\newcommand{\sumsub}[1]{\underset{#1}{\sum}}

\newcommand{\tricount}{\mathsf{TC}}
\newcommand{\VarObsRound}[0]{\overline{\Phi}_{\textnormal{obs}}}
\newcommand{\VarBobRound}[0]{\overline{\Phi}_{\textnormal{Bob}}}
\newcommand{\VarObs}[0]{\Phi_{\textnormal{obs}}}

\usepackage[english]{babel}
\usepackage[T1]{fontenc}
\usepackage{caption}

\newcommand{\eps}[0]{\varepsilon}

\usepackage[margin=1in]{geometry}

\algblockdefx{Withprob}{EndWithprob}[1]{\textbf{with probability} #1 \textbf{do}}{}

\makeatletter
\ifthenelse{\equal{\ALG@noend}{t}}%
  {\algtext*{EndWithprob}}
  {}%
\makeatother

\title{Gap-Majority Lemmas in Communication Complexity}
\author{ Pachara Sawettamalya\footnote{Department of Computer Science, Princeton University. Supported by NSF CAREER award CCF-233994. \ \href{mailto:ps3122@princeton.edu}{\url{ps3122@princeton.edu}}}  \and Huacheng Yu\footnote{Department of Computer Science, Princeton University. Supported by NSF CAREER award CCF-233994. \ \href{mailto:yuhch123@gmail.com}{\url{yuhch123@gmail.com}}}}
\date{\today}

\usepackage[backend=biber, style=alphabetic, doi=false, url=false, isbn=false, minalphanames = 3, maxalphanames = 4, maxbibnames=10]{biblatex}
\begin{document}
\maketitle

\begin{abstract}

We prove an information-theoretically optimal \emph{gap-majority lemma} in the two-player randomized communication model. For a base function $f: \mathcal{X} \to \{\pm 1\}$, its $n$-fold \emph{gap-majority composition}, denoted $\mathsf{GapMAJ} \circ f^n$, takes $n$ inputs $(X_1, \ldots, X_n)$ and distinguishes whether $f^{+n}(X_1,\ldots,X_n) := f(X_1) + \ldots + f(X_n)$ is at least $0.01\sqrt{n}$ or at most $-0.01\sqrt{n}$. We show that if computing $f$ with success probability $0.501$ requires $I$ bits of information, then computing $\mathsf{GapMAJ} \circ f^n$ with success probability $0.99$ requires $n \cdot (I - O(1))$ bits of information. This result is asymptotically optimal in two aspects: it achieves the correct linear scaling of information cost and the correct constant-constant tradeoff between error rates. This makes $\mathsf{GapMAJ}$, to our knowledge, only the third explicit outer gadget that admits a strong composition theorem in the two-player communication setting, following the identity and XOR gadgets.

From an application side, our gap-majority lemma can be viewed as a generic amplification tool that lifts the hardness of deciding $f$ into the hardness of approximating $f^{+n}$. Using this framework, we give a new proof to the communication lower bound of Gap-Hamming and derive a tight streaming lower bound of triangle counting, demonstrating the versatility of the gap-majority lemma.

% \vspace{2cm}

% \textcolor{red}{TODO LIST
% \begin{enumerate}
%     \item Title ok?
%     \item Abstract is not so smooth, especially the last sentence of both paragraphs.
%     \item Will remove table of contents later.
%     \item Is the intro smooth? Is there anything to add? Any missing references?
%     \item Is the last sentence of Section 1.1 (``This makes $\mathsf{GapMAJ}$ ...'') accurate? Is it too bold of a claim?
%     \item Is the second paragraph of Section 4 understandable?
%     \item Is the second \& third paragraphs of Section 5 understandable?
%     \item Does the name of Section 4 \& 5 make sense?
%     \item Is the structure of Section 6 ok?
%     \item For the arxiv version, let's add an acknowledgement for Elena, Yang, Yinchen, Santhoshini (for discussion) and for Sepehr (for triangle counting). For the submission, perhaps best not to include it.
%     \item I like the length of the current draft (28 pages once table of contents removed, actual content only 21). Can still add more if needed, but prefer to keep it under 30.
% \end{enumerate}
% }

\end{abstract}

\thispagestyle{empty}
\newpage
% \tableofcontents
\pagenumbering{roman}
\newpage
\pagenumbering{arabic}

\section{Introduction}
\label{sec:intro}

Given a base function $f$, consider the task of computing its $n$-fold product function $f^n(x_1, \ldots, x_n) := (f(x_1), \ldots, f(x_n)).$ A natural approach is to apply an optimal algorithm for $f$ independently to each input, using $n$ times the cost of computing a single instance. But is this optimal?

This problem is known as the \emph{direct-sum} question, and it is among the most central in computational complexity theory. It has been studied extensively across a variety of computational models such as circuit complexity \cite{Yao82, Lev87, Imp95,IW97}, communication complexity \cite{FederKNN95, BYJKS04, BBCR10, JPY12, BRWY13, MackenzieS25}, information complexity \cite{BBCR10, BR11, Braverman15}, and query complexity \cite{JainKS10, Dru12, Ben-DavidK18, BlaisB19, BenDavidB25, BesselmanGGM025},  to name just a few. 

In parallel, another line of research investigates the complexity of \emph{composition functions} $g \circ f^n$. For an arbitrary ``outer gadget'' $g$, the optimal bounds (i.e. proportional to the complexity of $f$ and $g$) are known under certain complexity measures such as deterministic query complexity \cite{Savick02, Tal13, Montanaro14} and quantum query complexity \cite{LeeMRSS11, Reichardt11a}. However, for general complexity measures, a starting point is to understand the complexity of $g\circ f^n$ when $g$ is an explicit gadget. 
That is, to study how the cost of computing $g \circ f^n$ scales with the cost of computing $f$ and with $n$, for a specific $g$. Naively, one could still compute $f$ on each input and then apply $g$ to the outputs. In the spirit of the direct-sum question, we may ask whether this simple strategy is optimal.

\begin{question}[Composition Theorem]
For a prescribed $n$-ary gadget $g$, how does the cost of computing $g \circ f^n$ scale with the cost of computing $f$, and with $n$?
\label{ques:gadget}
\end{question}

Besides the \emph{identity} gadget for which \Cref{ques:gadget} becomes the direct-sum problem, perhaps the most extensively studied gadget is the $\mathsf{XOR}$ function, which computes the parity of its $n$-bit input. Results of this type, often called the \emph{XOR lemma}, have also been well-established in various computational models, e.g. \cite{IW97, BBCR10, BKLS20, Yu22, IR24a, IR24b, SawettamalyaY25}. Beyond the XOR lemma, however, our understanding of \Cref{ques:gadget} is relatively limited and is mostly contained to query complexity with respect to the $\mathsf{OR}$, $\mathsf{GapOR}$, and $\mathsf{MAJ}$ gadgets \cite{FeigeRPU94, Ben-DavidGK020, GoosM21}.

\paragraph{Gap-Majority Lemma.} We focus on the \emph{gap-majority} gadget which computes the majority of $n$ binary inputs if the bias is $\Omega(\sqrt{n})$. Formally, define the gadget $\GapMAJ_n: \{\pm1\}^n \rightarrow \{\pm1\}$ as:
\begin{align*}
\GapMAJ_n(z_1,\ldots,z_n) :=
\begin{cases}
+1 & \text{if } \sum_{i=1}^{n} z_i \geq 0.01 \sqrt{n}, \\
-1 & \text{if } \sum_{i=1}^{n} z_i \leq -0.01 \sqrt{n}.
\end{cases}
\end{align*}
Given a base function $f: \mathcal{X} \rightarrow \{\pm 1\}$, we aim, with probability at least $0.99$, to compute the composite function
\begin{align*}
\GapMAJ_n \circ f^n(X_1,\ldots,X_n) :=
\begin{cases}
+1 & \text{if } \sum_{i=1}^{n} f(X_i) \geq 0.01 \sqrt{n},\\
-1 & \text{if } \sum_{i=1}^{n} f(X_i) \leq -0.01 \sqrt{n}.
\end{cases}
\end{align*}

At first glance, it is even unclear what the naive approach should be. One natural attempt is as follows: suppose there exists an algorithm $\mathcal{A}$ that computes $f(X)$ with error $1/(100n)$. We can estimate $\sum_{i=1}^n f(X_i)$ by applying $\mathcal{A}$ independently to each input and summing the results. By the union bound, this recovers the exact sum with probability $0.99$, yielding the correct output.

However, we can exploit the generous margin of $\pm 0.01\sqrt{n}$ to do better. Suppose $\mathcal{A}$ computes $f(X)$ with a small \emph{constant} error, say $10^{-100}$. By standard concentration bounds, the estimate of $\sum_{i=1}^n f(X_i)$ is within $\pm 10^{-10}\sqrt{n}$ of the true value with high constant probability, which suffices to distinguish between the ``$+1$'' and ``$-1$'' cases of $\GapMAJ_n \circ f^n$. Thus, the naive approach requires only $n$ times the resources needed to compute $f$ with low ``constant'' error.

With this benchmark established, the natural question is whether such linear scaling is optimal. We refer to this as the \emph{gap-majority lemma} question.

\begin{question}[Gap-Majority Lemma]
Does computing $\GapMAJ_n \circ f^n$ with success probability $0.99$ require $\Omega(n)$ times the cost of computing $f$ with success probability $0.501$?
\end{question}

\subsection{Main Results}

We study the gap-majority lemma in the realm of two-player randomized communication. In this setting, two players hold separated parts of inputs of a function $f: \mathcal{X} \times \mathcal{Y} \to \{\pm 1\}$, namely Alice has $X \in \mathcal{X}$, Bob has $Y \in \mathcal{Y}$, and they exchange a (possibly randomized) sequence of messages to compute $f(X, Y)$ with reasonable constant success probability. The scheme by which the players exchange messages and produce an output is called a \emph{communication protocol}.

Correspondingly, the \emph{gap-majority composition} problem of a base function $f$, denoted $\GapMAJ_n \circ f^n$, assigns $X = (X_1, \ldots, X_n)$ to Alice and $Y = (Y_1, \ldots, Y_n)$ to Bob. The players then must, with high constant probability, jointly compute
\begin{align*}
\GapMAJ_n \circ f^n(X_1,\ldots,X_n, Y_1,\ldots,Y_n) :=
\begin{cases}
+1 & \text{if } \sum_{i=1}^{n} f(X_i, Y_i) \geq 0.01 \sqrt{n}, \\
-1 & \text{if } \sum_{i=1}^{n} f(X_i, Y_i) \leq -0.01 \sqrt{n}.
\end{cases}
\end{align*}

Our main result establishes the gap-majority lemma in this setting when the notion of cost is \emph{internal information}, denoted $\IC(\cdot)$, which quantifies the amount of information revealed by the protocol.\footnote{We will formally define the information cost in \Cref{sec:prelim}.} In particular, we show that computing $\GapMAJ_n \circ f^n$ requires $\Omega(n)$ times the \emph{information} needed to compute a single instance of $f$, up to a small additive loss.

\begin{theorem}[Gap-Majority Lemma]
Let $f: \mathcal{X} \times \mathcal{Y} \rightarrow \{\pm 1\}$ and $\mu$ be an input distribution to $f$ such that $\E_\mu[f(x,y)] = 0$. Let $\pi$ be a protocol that computes $\GapMAJ_n \circ f^n$ correctly with probability $0.99$ over $\mu^n$. Then, there exists a protocol $\eta$ that computes $f$ correctly with probability $0.501$ over $\mu$ and $\IC(\eta) \le \frac{\IC(\pi)}{n} + O(1).$
\label{thm:gap_majority}
\end{theorem}

We make a few remarks regarding this theorem. First, we may even relax the range of $f$ so that it only requires $f(x,y) \in \{\pm 1\}$ over the support of $\mu$. Second, the $O(1)$-additive ``loss'' in information cost is in fact unavoidable. To see this, consider the simple example $f(x,y) := x - y$ with respect to the input distribution $\mu$ which is uniform over $(x,y) \in [10^4] \times [10^4]$ conditioned on $x - y \in \{\pm 1\}$. In this example, any protocol $\eta$ that computes $f(x,y)$ over $\mu$ with probability $0.501$ must reveal $\Omega(1)$ bits of information.\footnote{To see this, consider the protocol $\eta$ from Alice’s perspective. Knowing $x$, her prior view of $y$ is equally likely to be $x-1$ or $x+1$. After the communication, she gains a small constant advantage in predicting $y$. Hence, her posterior view of $y$ shifts from her prior view by $\Omega(1)$, meaning that she must have learned $\Omega(1)$ bits of information.} On the other hand, there is a deterministic protocol $\pi$ which computes $\GapMAJ_n \circ f^n$ while communicating only $O(\log{n})$ bits: Alice sends the value $\sum_{i=1}^n x_i$ to Bob, who then can compute $f^{+n}(x_1,\ldots,x_n,y_1,\ldots,y_n) = \sum_{i=1}^n x_i - \sum_{i=1}^n y_i$ and decides if it is at least $0.01\sqrt{n}$ or at most $-0.01\sqrt{n}$.

More importantly, although the $O(1)$-additive loss is small and oftentimes negligible, it comes at the expense of \emph{round complexity}: in our proof of \Cref{thm:gap_majority}, the protocol $\eta$ uses $O(\log n)$ more rounds than $\pi$ does. Such a mismatch in the number of rounds can be undesirable in certain applications. For example, a standard approach to proving single-pass streaming lower bounds is to reduce a streaming algorithm to a \emph{one-way} communication problem and then prove a communication lower bound. The one-way restriction is typically the primary source of hardness, and as a result, losing an additional $O(\log n)$ rounds in the reduction may lead to a significantly weaker lower bound.

To remedy this, we also provide a variant of the gap-majority lemma \emph{which also preserves the number of rounds!}

\begin{theorem}[Round-Preserving Gap-Majority Lemma]
Let $f: \mathcal{X} \times \mathcal{Y} \rightarrow \{\pm 1\}$ and $\mu$ be an input distribution to $f$ such that $\E_\mu[f(x,y)] = 0$. Let $\pi$ be an $r$-round protocol that computes $\GapMAJ_n \circ f^n$ correctly with probability $0.99$ over $\mu^n$. Then, there exists an $r$-round protocol $\eta$ that computes $f$ correctly with probability $0.501$ over $\mu$ and $\IC(\eta) \le \frac{\IC(\pi)}{n} + O(\log n).$
\label{thm:round_gap_majority}
\end{theorem}

It is helpful to view \Cref{thm:round_gap_majority} as a counterpart to \Cref{thm:gap_majority}: one can preserve the number of rounds at the cost of $O(\log n)$ extra bits, as opposed to $O(1)$ in the unbounded-round setting. It remains an open question whether the best of both worlds can be achieved; that is, whether there exists an $r$-round protocol $\eta$ with only $O(1)$ bits of information loss. We conjecture this is possible.

% It is helpful to view \Cref{thm:round_gap_majority} as a counterpart to \Cref{thm:gap_majority}: one can choose to either to preserve the number of rounds at the cost of $O(\log n)$ extra bits, or to lose only $O(1)$ bits for a price of extra $O(\log{n})$ rounds. It remains an open question whether the best of both worlds can be achieved; that is, whether one can construct such $r$-round protocol $\eta$ with only $O(1)$ bits of information loss. We conjecture this is possible. 

Finally, we emphasize that both of our gap-majority lemmas are \emph{strong}, in the sense that they achieve asymptotic optimality in both the scaling of information cost and the error tradeoffs, as witnessed by the naive approach. This makes $\GapMAJ$, to our knowledge, only the \emph{third} explicit gadget known to admit a strong composition theorem in two-player communication, following the identity gadget (i.e., the direct-sum theorem) \cite{BBCR10, BRWY13} and the XOR gadget (i.e., the XOR lemma) \cite{Yu22, IR24b, SawettamalyaY25}.

\subsection{Implication}
\label{subsec:framework}

\paragraph{Why Gap-Majority?} Gap-type problems have long served as fundamental tools for proving lower bounds for approximation problems. 
The underlying reason is simple yet powerful: many approximation tasks can be reduced to distinguishing between two well-separated cases of a decision problem—precisely the scenario highlighted by a ``gap'' function.

A particularly important instance of the gap-majority lemma arises when the base function is $f: \{\pm 1\} \times \{\pm 1\} \rightarrow \{\pm 1\}$ is the XOR function (with basis change), defined as $f(x,y) := xy$. In this case, the composition $\GapMAJ_n \circ f^n$ translates to the well-celebrated \emph{Gap-Hamming} problem in the two-player communication setting: Alice and Bob each hold a vector in $\{\pm 1\}^n$, and their goal is to decide whether the inner product between their vectors is at least $0.01\sqrt{n}$ or at most $-0.01\sqrt{n}$. The problem was formally introduced by Indyk and Woodruff \cite{IndykW03} who also proved the optimal lower bound of $\Omega(n)$ in \emph{one-way} communication. The complete \emph{unbounded round} lower bound, however, only arrived a decade later \cite{ChakrabartiR12,Vidick12,Sherstov12}. Regardless, since its inception (or even prior to that), Gap-Hamming has quickly become a cornerstone of communication complexity which, both explicitly and implicitly, leads to countless lower bounds across multiple computational models \cite{AlonMS99, IndykW03, Woodruff04, IndykW05, BlaisBM12, JayramW13, AndoniCKQWZ16}.

\paragraph{A new framework for proving lower bounds.}
While the Gap-Hamming problem alone has already yielded a rich array of lower bounds results, our gap-majority lemmas provide an even more generic framework for amplifying the hardness of ``any'' base \emph{decision} problem $f(x,y)$ into the hardness of an associated \emph{approximation} problem $f^{+n}(x_1,\ldots,x_n,y_1,\ldots,y_n) := f(x_1,y_1) + \ldots + f(x_n,y_n)$. By a standard reduction from our main results (\Cref{thm:gap_majority} and \ref{thm:round_gap_majority}), we obtain the following corollaries. 

\begin{corollary} Let $f: \mathcal{X} \times \mathcal{Y} \rightarrow \left\{\pm 1\right\}$ and $\mu$ be an input distribution to $f$ such that $\E_\mu[f(x,y)] = 0$. Let $\pi$ be a protocol that computes a $\pm 0.01 \sqrt{n}$ approximation to $f^{+n}$ with probability 0.99 over $\mu^n$. Then, there is a protocol $\eta$ which decides $f$ with probability 0.501 over $\mu$ and $\IC(\eta) \leq \frac{\IC(\pi)}{n}+O(1)$.
\label{cor:good_approx}
\end{corollary}

\begin{corollary} Let $f: \mathcal{X} \times \mathcal{Y} \rightarrow \left\{\pm 1\right\}$ and $\mu$ be an input distribution to $f$ such that $\E_\mu[f(x,y)] = 0$. Let $\pi$ be an $r$-round protocol that computes a $\pm 0.01 \sqrt{n}$ approximation to $f^{+n}$ with probability 0.99 over $\mu^n$. Then, there is an $r$-round protocol $\eta$ which decides $f$ with probability 0.501 over $\mu$ and $\IC(\eta) \leq \frac{\IC(\pi)}{n}+O(\log n)$.
\label{cor:good_approx_1way}
\end{corollary}
These corollaries yield a clean framework for proving communication lower bounds for approximation problems. Conceptually, this framework plays a role analogous to that of direct-sum theorems, which lift lower bounds from a base problem to its $n$-fold product version, a technique that has seen frequent success in the literature. In our setting, however, the emphasis is on systematically lifting hardness from a \emph{decision} problem to an \emph{approximation} problem via the following recipe:
\begin{enumerate}
    \item \textbf{Reduction.} Show that an algorithm for approximating the target problem $\mathcal{P}$ can be used to approximate $f^{+n}$ within $\pm 0.01\sqrt{n}$, for some appropriate choices of a base decision problem $f$ and $n \in \mathbb{Z}^+$.
    \item \textbf{Base lower bound.} Prove that any protocol that decides $f$ with probability $0.501$ requires at least $\mathcal{I}$ bits of information.
\end{enumerate}
Then, by applying either \Cref{cor:good_approx} or \Cref{cor:good_approx_1way} as a black box, 
we obtain an information and communication lower bound of roughly $\Omega(n \cdot \mathcal{I})$ for the approximation problem $\mathcal{P}$.

\paragraph{Lower Bound Applications.}
\label{subsec:intro_apps}

For completeness, we showcase the power of the gap-majority lemma through a few of its quick applications.

\begin{itemize}
    \item \textbf{Communication Complexity of Gap-Hamming.} As previously noted, the Gap-Hamming problem is a particularly important instance of the gap-majority composition. While the optimal communication lower bound of $\Omega(n)$ bits has long been established \cite{ChakrabartiR12, Vidick12, Sherstov12}, it is only natural to recover such bound through the lens of our gap-majority framework.  We present this proof in \Cref{subsec:GH}. Indeed, treating the gap-majority lemma as a ``black-box'' amplification tool, our proof proceeds through a very simple and clean reduction to the classical \emph{set-disjointness} problem. 
    
    % In fact, our proof is entirely information-theoretic and indeed yields an $\Omega(n)$ \emph{information complexity} lower bound, thereby joining the list of such results in the literature \cite{ChakrabartiKW12, KerenidisLLRX15, BravermanGPW16}.
    
    \item \textbf{Streaming Complexity of Triangle Counting.} Given an $n$-vertex, $m$-edge graph with $T$ triangles presented as an edge-arrival stream, the \emph{triangle counting} problem asks for a $(1 \pm \eps)$-approximation to the number of triangles in $G$, that is, to output an estimate $\widetilde{T}$ satisfying $|\widetilde{T}-T|\leq \eps T$. A naive solution is to store the entire graph and use $\widetilde{O}(m)$ space. This bound turns out to be beatable subjected to the range of parameters. For instance, Buriol et al. \cite{BuriolFLMS06} gave a sampling-based algorithm using $\widetilde{O}\left(\frac{mn}{\eps^2 T}\right)$ space, surpassing the naive algorithm when $G$ has many triangles. 
    
    % This bound was subsequently matched by Ahn, Guha, and McGregor~\cite{AhnGM12} who further extended the result to turnstile streams, allowing both edge insertions and deletions.
    
    To this end, we show, via gap-majority framework, that these two quantitative upper bounds are jointly optimal: any randomized single-pass streaming algorithm for $(1 \pm \eps)$-approximate triangle counting requires $\Omega\left(\min\left\{m,\;\frac{mn}{\eps^2 T}\right\}\right)$ bits of space, for a reasonable range of $T$. Prior to our work, the best-known space lower bound is $\Omega\left(\min\left\{m,\;\frac{mn}{T}\right\}\right)$ bits due to Braverman, Ostrovsky, and Vilenchik \cite{BravermanOV13}.\footnote{In fact, there are other known lower bounds parameterized by additional structural parameters. We omit these bounds since our regime is parameterized solely by $n,m,T$, and $\eps$. Readers may consult \cite{JayaramK21} for further references.} Hence, our contribution lies in bringing up the missing $1/\eps^2$ factor of these existing lower bounds. We present this proof in \Cref{subsec:tricount}.
\end{itemize}

% Moreover, with nontrivial effort, one could separately establish lower bounds of $\Omega\left(\frac{mn}{T}\right)$ and $\Omega(1/\eps^2)$ for appropriate ranges of $n,m,T,$ and $\eps$.

For brevity, we defer the extended introduction and complete lower bound proof of both problems to \Cref{sec:lb_apps}. We leave it as an intriguing direction whether our gap-majority lemmas can be applied (either in a black-box or white-box manner) to derive lower bounds beyond those listed here.

\section{Preliminaries}
\label{sec:prelim}

\paragraph{Notations.} For a base function $f$ with real-valued outputs, we define its \emph{$n$-fold summation} as $f^{+n}(X_1,\ldots,X_n) := f(X_1) + \ldots + f(X_n)$. We may occasionally write $f^{+n}$ simply as $f^{+}$, omitting the number of summands. 

Following standard convention in probability, we use uppercase letters to denote random variables and lowercase letters for their realization. For a joint distribution $\pi$ over multiple random variables, $\pi(A)$ denotes the marginal distribution of $A$, and $\pi(a)$ its probability mass at $A=a$. Similarly, $\pi(A \mid B)$ denotes the conditional distribution of $A$ given $B$, and $\pi(A \mid b)$ the distribution of $A$ conditioned on $B=b$. We write $\E[A]$ and $\Var[A]$ for the expectation and variance of $A$, respectively, and use subscripts to indicate the variable of integration, e.g. $\E_A[f(A)]$.  
When clear from context, we omit the subscript. We will also frequently use the following simple fact which can be proven via the Cauchy-Schwarz inequality.

% The following simple fact will be frequently used in our proofs.

% In addition, we will frequently rely on the following simple fact.

\begin{fact}
    $\Esub{B} \ \Varsub{A} \left[f(A)\mid B\right] = \Esub{A}\left[f(A)^2\right] - \Esub{B}\left[\Esub{A}\left[f(A)\mid B\right]^2\right] \leq \Varsub{A}\left[f(A)\right].$
\label{fact:cond_var}
\end{fact}

\paragraph{Communication Protocol.}
We adopt the distributional view of randomized communication protocols.  
A protocol $\pi$ is represented as a joint distribution $(X,Y,M,R)$, where $(X,Y) \sim \mu$ is the input pair given to Alice and Bob, $M= (M_1,M_2,\ldots)$ is the transcript, and $R$ denotes public randomness. The protocol operates as follows. Alice and Bob receive inputs $(X,Y) \sim \mu$ respectively. Assume Alice speaks first. She draws the first message $M_1 \sim \pi(M_1 \mid XR)$ and sends it to Bob. Bob then draws the second message $M_2 \sim \pi(M_2 \mid M_1YR)$ and sends it back to Alice. The communication between players goes on as needed, and thereby forms a transcript $M$. Without restricting the number of \emph{rounds} (i.e. the number of messages), if the protocol is designated to compute a function $f$, we may assume that the final message of the protocol is the players' answer to $f(X,Y)$. We can always enforce this assumption, e.g. if Alice is responsible for computing $f(X,Y)$, she simply sends an extra (short) message to Bob indicating the answer. The \emph{communication cost} of $\pi$, denoted $\mathsf{CC}(\pi)$, is measured by the worst-case length of the message $M$.

% We will also examine protocols that operates on product inputs which Alice receives an $n$-coordinate input $X = (X_1,\ldots,X_n)$ and Bob receives an $n$-coordinate input $Y = (Y_1,\ldots,Y_n)$ such that each $(X_i,Y_i) \sim \mu$ i.i.d.

\paragraph{Protocol with Bounded Rounds.} A special case of communication protocols occurs when we restrict the number of rounds of communication to $r$. We may assume wlog that $r$ is odd. In this setting, Alice and Bob can communicate for $r$ rounds and towards the end of their conversation, Bob is asked to produce an answer to $f(X,Y)$. Note that Alice may not know this answer. We denote an $r$-round protocol by $\pi = (X,Y,M,A,R)$, where $M = (M_1,\ldots,M_r)$ is the sequence of $r$ messages, $R$ is shared public randomness, and $A$ is Bob’s output computed from $MYR$. Denote $\mathcal{A}$ to be the space of all possible answers $A$. For instance, we use $\mathcal{A} = \{\pm 1\}$ for decision problems.

\paragraph{Information Cost.} The information cost of a protocol $\pi = (X,Y,M,R)$ is given by the quantity:
$$\IC(\pi) := I(M:X \mid YR) + I(M:Y \mid XR)$$
where $I(\cdot:\cdot)$ denotes mutual information. It is helpful to interpret $\mathsf{IC}(\pi)$ as the amount of \emph{information} that $\pi$ reveals about the players' inputs. To see this, from Alice's view at the end of $\pi$, she learns the transcript $M$ while already knowing her own input $X$ and randomness $R$. Hence, the term $I(M:Y \mid XR)$ quantifies the information she gains about Bob's input $Y$ after running $\pi$. 

The information cost is also defined by the same quantity even if $\pi = (X,Y,M,A,R)$ is bounded-round. Note that we always have $\mathsf{IC}(\pi) \leq \mathsf{CC}(\pi)$.

\paragraph{Rectangle Property.} Rectangle property is one of the defining features of communication protocols. It asserts that in a communication protocol whose inputs consist of multiple coordinates drawn from a product distribution, conditioned on a set of disjoint inputs and the message $M$ (and randomness $R$), the remaining portions of players' input are independent.

\begin{proposition}[Rectangle property; Section 3.4 of \cite{Yu22}] Let $\pi = (X,Y,M,R)$ be a protocol over $n$-coordinate inputs $(X,Y) \sim \mu^n$. For any partition of coordinate $[n] = P \cup Q$, we have $X_P \perp Y_Q \mid MRX_QY_P$. If $\pi = (X,Y,M,A,R)$ is a bounded-round protocol, then $X_P \perp Y_Q \mid MRAX_QY_P$.
\label{prop:rectangle}
\end{proposition}

% For more detailed description and discussions of communication protocols, information cost, the rectangle property, and more, see Section~5 of \cite{SawettamalyaY25}.

\section{Approach}
\label{sec:approach}

In this section, we outline the roadmap toward our two main results (\Cref{thm:gap_majority} and \Cref{thm:round_gap_majority}), with the proof of key technical lemmas deferred to later sections.

\subsection{Gap-Majority Lemma} 

Our proof of \Cref{thm:gap_majority} proceeds by introducing a new measure called the \emph{observer’s variance}.

\begin{definition}[Observer's variance] Let $\pi = (X,Y,M,R)$ be a communication protocol, and let $f$ be a function on the input $(X,Y)$.  
The \emph{observer’s variance} of $\pi$ with respect to $f$ is defined as
$$\VarObs(\pi @ f) := \Esub{MR} \ \Varsub{XY} \ [f(X,Y) \mid MR].$$
\end{definition}

In words, the observer’s variance measures, \emph{from the perspective of an external observer}, the uncertainty of $f(X,Y)$ after the protocol concludes. Particularly from an external view, they see only the transcript $M$ and the public randomness $R$, but not the players’ private inputs.  For each realization of $MR$, the observer considers the posterior distribution of $f(X,Y)$ conditioned on these values and evaluates the uncertainty $f(X,Y)$ via its variance.  Averaging over all $MR$ then yields $\VarObs(\pi @ f)$, representing the expected variance of $f$ given the protocol’s external view.

With this definition in place, we now state the two key ingredients that together establish the gap-majority lemma (\Cref{thm:gap_majority}).

\paragraph{Ingredient \#1: Gap-majority implies low variance.} 
We first show that any protocol that computes gap-majority with high constant probability must admit low linear observer's variance, say at most $0.99n$. This phenomenon was previously observed in a streaming setting by Braverman, Garg, and Woodruff~\cite{BravermanGW20}, who showed that any streaming algorithm capable of computing the exact majority necessarily has low linear variance. For our purposes, we strengthen this result to the communication setting for gap-majority composition, as stated in \Cref{lem:gapmaj_then_low_variance}. We emphasize that our proof follows closely from \cite{BravermanGW20} to which we defer to the appendix.

\begin{restatable}[Gap-majority implies low variance]{lem}{variancebound}
Let $f: \mathcal{X} \times \mathcal{Y} \to \{\pm 1\}$ and let $\mu$ be an input distribution to $f$ satisfying $\E_\mu[f(x,y)] = 0$. Suppose $\pi$ is a protocol that computes $\GapMAJ_n \circ f^n$ with probability $0.99$ over $\mu^n$. Then $\VarObs(\pi @ f^{+n}) \le 0.99n$.
\label{lem:gapmaj_then_low_variance}
\end{restatable}

\paragraph{Ingredient \#2: A new direct-sum theorem.} 
To motivate our second ingredient, it is best to discuss the following toy question: How should the information cost and observer's variance for $f^{+n}$ scale relative to those for a single instance of $f$?

In an easy direction, suppose that we have a protocol $\theta$ for computing $f$ using $\mathcal{I}$ bits of information, and with variance $\sigma^2$. We can then employ a naive approach by applying $\theta$ on the $n$ pairs of inputs independently, and then sum up the answers. This approach simply use $n \cdot \mathcal{I}$ bits of information. But how does the variance of $f^{+n}$ scale? The key observation is that each execution of $\theta$ is independent; thus, we can directly sum up the variance to $n\sigma^2$. Notice that the optimality question of this naive protocol is exactly the direct-sum problem: we may ask if the scaling factor $n$ of information costs and variances are tight, or even stronger, are they \emph{simultaneously} tight?

Our second ingredient is a direct-sum theorem that answers this question in an affirmative, barring a small additive loss. We state the result here and present its proof in \Cref{sec:direct_sum_var}.

\begin{restatable}[Main technical lemma]{lem}{mainlemma}
Let $f: \mathcal{X} \times \mathcal{Y} \rightarrow \{\pm 1\}$ be an arbitrary function, and let $\mu$ be any input distribution over $\mathcal{X} \times \mathcal{Y}$. Suppose $\pi$ is a protocol over the input distribution $\mu^n$. Then, for any $\Delta \in (0,1)$, there exists a protocol $\eta$ over the input distribution $\mu$ such that 
$$\IC(\eta) \leq \frac{\IC(\pi)}{n} + O\Big(1 + \log \frac{1}{\Delta}\Big) \hspace{5mm} \text{ and } \hspace{5mm}
    \VarObs(\eta @ f) \leq \frac{\VarObs(\pi @ f^{+n})}{n} + \Delta.$$
\label{lem:main_lemma}
\end{restatable}
\vspace{-3.5mm}

Crucially, we view this lemma as a central technical contribution of our work. Note that unlike \Cref{lem:gapmaj_then_low_variance}, the condition $\E_\mu[f(x,y)] = 0$ is \emph{not} required here. Whether \Cref{lem:main_lemma} has further applications beyond the gap-majority setting, in our opinion, is a direction worth exploring. 

As a remark, \Cref{lem:main_lemma} extends naturally to the setting where $f$ maps to the interval $[-d,d]$ and $\Delta \in (0,d^2)$ for some $d \in \mathbb{R}^+$. In this case, the statement remains the same except that the additive loss in the information cost becomes $O(1+\log{\frac{d^2}{\Delta}})$. For brevity, we omit this proof.

\paragraph{Putting all together.} We now combine both ingredients to recover the gap-majority lemma. 

\begin{proof}[Proof of \Cref{thm:gap_majority}] Recall that $\pi$ is a protocol that computes $\GapMAJ_n \circ f^n$ over $\mu^n$ with probability $0.99$. By \Cref{lem:gapmaj_then_low_variance}, we have $\VarObs(\pi @ f^{+n}) \leq 0.99n.$ Then, applying \Cref{lem:main_lemma} with $\Delta = 0.008$ yields a protocol $\eta = (X,Y,M,R)$ over inputs $(X,Y) \sim \mu$ such that $\IC(\eta) \leq \frac{\IC(\pi)}{n}+O(1)$ and 
$\VarObs(\eta @ f) \leq \frac{ \VarObs(\pi @ f^{+n})}{n} + 0.008 \leq 0.998$. Expanding the observer's variance, we have 
$$\Esub{MR} \ \Varsub{XY} \ [f(X,Y) \mid MR] \leq 0.998.$$

Now consider the following protocol over inputs $(X,Y) \sim \mu$: the players first execute $\eta$, and then outputs their answer to $f(X,Y)$ as the more-likely answer between $\{\pm 1\}$ among the posterior distribution $\eta(f(X,Y) \mid MR)$. The information cost of this protocol remains at $\IC(\eta) \leq \frac{\IC(\pi)}{n}+O(1)$.

It only remains to argue the correctness probability of such protocol. With probability $\eta(mr)$, the protocol realizes $MR = mr$. Denote $0 \leq p_{mr} \leq 1/2$ be such that the distribution of $\eta(f(X,Y) \mid mr)$ is $\{\pm 1\}$ with probabilities $\{p_{mr},1-p_{mr}\}$. Then, we have
$$0.998 \geq \Esub{MR} \ \Varsub{XY} \ [f(X,Y) \mid MR] = \sum_{mr} \eta(mr) \cdot 4p_{mr}(1-p_{mr}) \geq 2 \cdot \sum_{mr} \eta(mr) \cdot p_{mr}.$$
Meanwhile, the distributional error is $\sumsub{mr} \ \eta(mr) \cdot p_{mr} \leq 0.499.$ In other words, the protocol produces a correct answer with probability $0.501$. 
\end{proof}

\subsection{Round-Preserving Gap-Majority Lemma}
\label{subsec:approach_oneway}

We now sketch the proof of \Cref{thm:round_gap_majority} which is an analogue of \Cref{thm:gap_majority} restricted to the class of $r$-round communication. As outlined in \Cref{sec:prelim}, we may assume that $r$ is odd and write an $r$-round protocol as $\pi = (X,Y,M,A,R)$, where $M = (M_1,\ldots,M_r)$ is the sequence of $r$ messages, and $A$ is the output produced solely by Bob after the $r$ rounds of communication. Since Bob can no longer communicate the answer $A$ back to Alice (otherwise it would have taken an extra round), the observer's variance is no longer a reliable proxy for distributional error. To accommodate this, we refine the notion of variances by distinguishing it into two types: observer's and Bob's.

\begin{definition}[Variances in $r$-round communication] \label{def:1-way-var} Given an $r$-round protocol $\pi = (X,Y,M,A,R)$ and a function $f$ over input pair $(X,Y)$. Define the \emph{observer's variance} of $\pi$ with respect to $f$ by the quantity:
$$\VarObsRound(\pi @ f) := \Esub{MRA} \ \Varsub{XY} \ [f(X,Y) \mid MRA]$$
and \emph{Bob's variance} of $\pi$ with respect to $f$ by the the quantity:
    $$\VarBobRound(\pi @ f) := \Esub{MRAY} \ \Varsub{X} \ [f(X,Y) \mid MRAY].$$
\end{definition}

In principle, the observer's variance remains the same as before, only now adapted to the bounded round setting: an external observer sees $MRA$ and measures the expected variance of $f(X,Y)$ conditioned on this view. Bob's variance, however, is new: it quantifies the expected variance of $f(X,Y)$ from Bob's perspective at the end of the protocol, given that he sees $MRAY$. When $f$ is a decision problem, this quantity turns out to be related to the distributional error.

Following the same avenue, we introduce the two key ingredients tailored to $r$-round communication. The variance bound (\Cref{lem:gapmaj_then_low_variance_1way}) parallels \Cref{lem:gapmaj_then_low_variance}, and for brevity, we omit its proof. The round-preserving direct-sum theorem (\Cref{lem:main_lemma_oneway}) is analogous to \Cref{lem:main_lemma}, except for a slightly larger additive loss in the information cost. Notably, the proof of \Cref{lem:main_lemma_oneway} is more technically involved than that of \Cref{lem:main_lemma} due to the round constraint. We present this proof in \Cref{sec:direct_sum_oneway}.

\begin{restatable}[Bounded-round version of \Cref{lem:gapmaj_then_low_variance}]{lem}{varianceboundoneway} Let $f: \mathcal{X} \times \mathcal{Y} \to \{\pm 1\}$ and let $\mu$ be an input distribution to $f$ satisfying $\E_\mu[f(x,y)] = 0$.  Suppose $\pi$ is an $r$-round protocol that computes $\GapMAJ_n \circ f^n$ with probability $0.99$ over $\mu^n$. Then $\VarObsRound(\pi @ f^{+n}) \leq 0.99n$.
\label{lem:gapmaj_then_low_variance_1way}
\end{restatable} 

\begin{restatable}[Main technical lemma for $r$-round communication]{lem}{mainlemmaoneway}
Let $f: \mathcal{X} \times \mathcal{Y} \rightarrow \{\pm 1\}$ be an arbitrary function, and let $\mu$ be any input distribution over $\mathcal{X} \times \mathcal{Y}$. Suppose $\pi$ is an $r$-round protocol over the input distribution $\mu^n$ which Bob produces an answer from a discrete set $\A$. Then, for any $\Delta \in (0,1)$, there is an $r$-round protocol $\eta$ for computing $f$ over the input distribution $\mu$ such that 
\begin{equation*}
    \IC(\eta) \leq \frac{\IC(\pi)}{n}+O(|\A|\cdot \log{\frac{n}{\Delta}}) \hspace{7mm} \text{ and } \hspace{7mm}\VarBobRound(\eta @ f) \leq \frac{\VarObsRound(\pi @ f^{+n})}{n} + \Delta.
\end{equation*}
\label{lem:main_lemma_oneway}
\end{restatable}
\vspace{-1.5mm}

We conclude by combining \Cref{lem:gapmaj_then_low_variance_1way} and \Cref{lem:main_lemma_oneway} to prove our second main result.

\begin{proof}[Proof of \Cref{thm:round_gap_majority}] Recall that $\pi$ is an $r$-round protocol that computes $\GapMAJ_n \circ f^n$ over $\mu^n$ with probability $0.99$. By \Cref{lem:gapmaj_then_low_variance_1way}, we have $\VarObsRound(\pi @ f^{+n}) \leq 0.99n$. Then, applying \Cref{lem:main_lemma_oneway} with $\Delta = 0.008$ yields an $r$-round protocol $\eta = (X,Y,M,A,R)$ over inputs $(X,Y) \sim \mu$ such that $\IC(\eta) \leq \frac{\IC(\pi)}{n} + O(\log{n})$ and $\VarBobRound(\eta @ f) \leq \frac{\VarBobRound(\pi @ f^{+n})}{n} + 0.008 \leq 0.998$. Here, we use the fact that the answer space is $\A = \{\pm 1\}$ (i.e. output to $\GapMAJ_n \circ f^n$). Expanding Bob's variance, we have
$$\Esub{MRAY} \ \Varsub{X} \ [f(X,Y) \mid MRAY] \leq 0.998.$$

Now consider the following $r$-round protocol over inputs $(X,Y) \sim \mu$: the players first execute $\eta$, then Bob outputs $f(X,Y)$ as the more-likely answer between $\{\pm1\}$ among the posterior distribution $\eta(f(X,Y) \mid MRAY)$. The information cost of this protocol remains at $\IC(\eta) \leq \frac{\IC(\pi)}{n} + O(\log{n})$.

It only remains to argue the correctness probability of such protocol. With probability $\eta(mray)$, the protocol realizes $MRAY=mray$. Denote $0 \leq p_{mray} \leq 1/2$ be such that the distribution $\eta(f(X,Y) \mid MRAY)$ is $\{\pm 1\}$ with probabilities $\{ p_{mray}, 1- p_{mray}\}$. Then, we have 
$$0.998 \geq \Esub{MRAY} \ \Varsub{X} \ [f(X,Y) \mid MRAY] = \sum_{mray} \eta(mray) \cdot 4p_{mray}(1-p_{mray}) \geq 2 \cdot \sum_{mray} \eta(mray) \cdot p_{mray}.$$
Meanwhile, the distributional error is $\sumsub{mray} \eta(mray) \cdot p_{mray} \leq 0.499.$ In other words, the protocol produces a correct answer with probability at least $0.501$.
\end{proof}

\subsection{Proof Outline}

Three proofs are deferred from this section. The proof of \Cref{lem:gapmaj_then_low_variance} is provided in Appendix \ref{appendix:BGW_reprove}. The proof of \Cref{lem:main_lemma} and \Cref{lem:main_lemma_oneway} are provided in \Cref{sec:direct_sum_var} and \Cref{sec:direct_sum_oneway} respectively.

\section{Direct-Sum for Variance and Information}
\label{sec:direct_sum_var}

In this section, we prove \Cref{lem:main_lemma}, thereby completing the proof of the gap-majority lemma (\Cref{thm:gap_majority}). We adopt the setting of the lemma: let $f: \mathcal{X} \times \mathcal{Y} \rightarrow \{\pm 1\}$ be an arbitrary base function, and let $\mu$ be an arbitrary input distribution. Denote by $\pi = (X,Y,M,R)$ a protocol over the product distribution $\mu^n$, where $X = (X_1,\ldots,X_n)$ and $Y = (Y_1,\ldots,Y_n)$ with each $(X_i,Y_i)$ drawn independently from $\mu$. Let $\Delta \in (0,1)$ be an arbitrary additive loss parameter.

A key component of our proof is an adaptation of the \emph{protocol decomposition} technique \cite{Yu22, SawettamalyaY25}. Given a protocol $\pi$ operating on $n$ input pairs drawn according to $\mu^n$, a \emph{decomposition} is a procedure that yields two smaller-yet-meaningful subprotocols, $\pi_0$ and $\pi_1$, each operating on $n/2$ input pairs drawn according to $\mu^{n/2}$. As employed in \cite{Yu22, SawettamalyaY25}, their decompositions are carefully designed to trace how \emph{information} is distributed across different parts of the input. A key distinction in our work is that our decomposition is tailored to handle not only the distribution of information, but also the distribution of \emph{observer's variances} across all the $n$ input pairs.

Particularly, we describe the protocols $\pi_0$ and $\pi_1$ (yielded by our decomposition) in \Cref{pi0,pi1}. Each protocol operates on $n/2$ pairs of input, denoted $(X', Y')$, drawn from $\mu^{n/2}$.

\begin{figure}[H]
    \centering
    \begin{tabular}{|p{15.7cm}|}
    \hline  ~\\
    \multicolumn{1}{|c|}{\textbf{A protocol $\pi_0$ over the input distribution $(X',Y') \sim \mu^{n/2}$}} \\ 
    \\ \hline
    \vspace{-2.85mm}
    % \textbf{Procedures:}
    \begin{enumerate}
        \item Alice pretends that her input consists of $n$ coordinates, and writes it as $X = (X_0,X_1)$ where $X_0$ and $X_1$ denotes the first and last $n/2$  coordinates of $X$ respectively. 
        \vspace{-1.1mm}
        \item Bob pretends that his input consists of $n$ coordinates, and writes it as $Y = (Y_0,Y_1)$ where $Y_0$ and $Y_1$ denotes the first and last $n/2$ coordinates of $Y$ respectively.
        \vspace{-1.1mm}
        \item Alice embeds $X'$ into $X_0$, and Bob embeds $Y'$ into $Y_0$.
        \vspace{-1.1mm}
        \item Players use public randomness to jointly draw $Y_1$. Alice further privately draws $X_1$.
        \vspace{-1.1mm}
        \item Alice and Bob simulate $\pi$ with respect to the inputs $(X,Y)$.
        \vspace{-1.1mm}
        \item Upon completion, Alice computes $Z_0 = \Esub{Y_0}\left[f^+(X_0, Y_0) \mid MRX_0Y_1\right]$
        and sends it to Bob.
    \end{enumerate}
    \vspace{0mm}
    \\
    \hline
    \end{tabular}
    \caption{The protocol $\pi_0$ resulting from decomposing $\pi$.}
    \label{pi0}
\end{figure}

\begin{figure}[H]
    \centering
    \begin{tabular}{|p{15.7cm}|}
    \hline ~\\
    \multicolumn{1}{|c|}{\textbf{A protocol $\pi_1$ over the input distribution $(X',Y') \sim \mu^{n/2}$}} \\ 
    \\ \hline
    \vspace{-2.85mm}
    % \textbf{Procedures:}
    \begin{enumerate}
        \item Alice pretends that her input consists of $n$ coordinates, and writes it as $X = (X_0,X_1)$ where $X_0$ and $X_1$ denotes the first and last $n/2$ coordinates of $X$ respectively. 
        \vspace{-1.1mm}
        \item Bob pretends that his input consists of $n$ coordinates, and writes it as $Y = (Y_0,Y_1)$ where $Y_0$ and $Y_1$ denotes the first and last $n/2$ coordinates of $Y$ respectively. 
        \vspace{-1.1mm}
        \item Alice embeds $X'$ into $X_1$, and Bob embeds $Y'$ into $Y_1$.
        \vspace{-1.1mm}
        \item Players use public randomness to jointly draw $X_0$. Bob further privately draws $Y_0$.
        \vspace{-1.1mm}
        \item Alice and Bob simulate $\pi$ with respect to the inputs $(X,Y)$.
        \vspace{-1.1mm}
        \item Upon completion, Bob computes $Z_1 = \Esub{X_1}\left[f^+(X_1, Y_1) \mid MRX_0Y_1\right]$
        and sends it to Alice.
    \end{enumerate}
    \vspace{0mm}
    \\
    \hline
    \end{tabular}
    \caption{The protocol $\pi_1$ resulting from decomposing $\pi$.}
    \label{pi1}
\end{figure}

Equivalently, we may adopt the following distributional view of each protocol. Note that we write $(M,Z_i)$ to denote a transcript $M$ followed by an extra message $Z_i$.

\begin{itemize}
    \item $\pi_0 = (X_0, Y_0, (M,Z_0), RY_1)$ where $Z_0 = \Esub{Y_0}\left[f^+(X_0, Y_0) \mid MRX_0Y_1\right]$ computable by Alice. This is because $Z_0$ is a deterministic function of $MRX_0Y_1$ in which she knows by the end of $\pi$.
    \item $\pi_1 = (X_1, Y_1, (M,Z_1), RX_0)$ where $Z_1 = \Esub{X_1}\left[f^+(X_1, Y_1) \mid MRX_0Y_1\right]$ computable by Bob. This is because $Z_1$ is a deterministic function of $MRX_0Y_1$ in which he knows by the end of $\pi$.
\end{itemize}

For the moment, we assume that Alice and Bob can communicate the exact values of $Z_0$ and $Z_1$, and that these variables take values from some discrete space $\mathcal{Z}$. Of course, this assumption does not necessarily hold in general. The purpose of this simplification is purely expository; we will remove this assumption shortly, incurring only minor losses.

The primary advantage of this decomposition is that it enables a clean algebraic decomposition of our two key measures of interest: the information cost and the observer's variance.

\paragraph{Decomposition of Information Costs.} We first show that under the decomposition, the information costs decompose linearly with some small additive loss.
\begin{claim} $\IC(\pi_0) + \IC(\pi_1) \leq \IC(\pi) + 4 \cdot \log_2{|\mathcal{Z}|}$.
\label{clm:IC_decomp}
\end{claim}
\begin{proof} Recall that $\pi_0 = (X_0, Y_0, (M,Z_0), RY_1)$. We can then write:
\begin{align*}
\IC(\pi_0) 
&= I(MZ_0 : X_0 \mid RY) + I(MZ_0 : Y_0 \mid RX_0Y_1) \\
&\le I(M : X_0 \mid RY) + I(M : Y_0 \mid RX_0Y_1) + 2 \log_2 |\mathcal{Z}| \\
&= I(M : X_0 \mid RY) + I(M : Y_0 \mid RXY_1) + 2 \log_2 |\mathcal{Z}|, \tag{rectangle property}
\end{align*}
where the last equality follows because $X_1 \perp Y_0 \mid MRX_0Y_1$.

Similarly, we also have $\IC(\pi_1) \le I(M : X_1 \mid RX_0Y) + I(M : Y_1 \mid RX) + 2 \log_2 |\mathcal{Z}|.$ Adding the two bounds and applying the chain rule for mutual information yields $\IC(\pi_0) + \IC(\pi_1) 
\le I(M : X \mid RY) + I(M : Y \mid RX) + 4 \log_2 |\mathcal{Z}| 
= \IC(\pi) + 4 \log_2 |\mathcal{Z}|$.
\end{proof}

\paragraph{Decomposition of Variances.} Next, we show that the observer's variance decomposes linearly. We begin by reinterpreting the observer's variance of each protocol.
\begin{claim} We have  
$$\VarObs(\pi_0 @ f^{+n/2}) = \Esub{MRX_0Y_1} \ \Varsub{Y_0}\left[f^+(X_0,Y_0) \mid MRX_0Y_1\right]$$ and $$\VarObs(\pi_1 @ f^{+n/2}) = \Esub{MRX_0Y_1} \ \Varsub{X_1}\left[f^+(X_1,Y_1) \mid MRX_0Y_1\right].$$
\label{clm:var_pi0_pi1}
\end{claim} 
\begin{proof} By symmetry, we will only derive the first equation. Recall that $\pi_0 = (X_0, Y_0, (M,Z_0), RY_1)$. We first expand
\begin{align}
    \VarObs(\pi_0 @ f^{+n/2}) & = \Esub{MRY_1Z_0} \ \Varsub{X_0Y_0} \left[ f^{+}(X_0, Y_0)\mid MRY_1Z_0\right] \nonumber \\
    & = \Esub{X_0Y_0} \left[f^+(X_0, Y_0)^2\right] - \sum_{mry_1z_0} \pi(mry_1z_0) \cdot \left[\Esub{X_0Y_0} \left[f^+(X_0, Y_0) \mid mry_1z_0 \right]\right]^2 \label{eq:expand_gamma0}
\end{align}
We further unpack the final expectation term:
\begin{align*}
    \Esub{X_0Y_0} \left[f^+(X_0, Y_0) \mid mry_1z_0\right] & = \sum_{x_0} \pi(x_0 \mid mry_1z_0) \cdot \Esub{Y_0} \left[f^+(x_0, Y_0) \mid mrx_0y_1z_0\right] \\
    & = \sum_{x_0} \pi(x_0 \mid mry_1z_0) \cdot \Esub{Y_0} \left[f^+(x_0, Y_0) \mid mrx_0y_1\right] \tag{$z_0$ is deterministic given $mrx_0y_1$} \\
    & = \sum_{x_0} \pi(x_0 \mid mry_1z_0) \cdot z_0 \tag{definition of $z_0$} \\
    & = z_0.
\end{align*}
Plugging this back into \Cref{eq:expand_gamma0}, we obtain
\begin{align}
    \VarObs(\pi_0 @ f^{+n/2})
    & = \Esub{X_0Y_0} \left[f^+(X_0, Y_0)^2\right] - \sum_{mry_1z_0} \pi(mry_1z_0) \cdot z_0^2 \nonumber \\
    & = \Esub{X_0Y_0} \left[f^+(X_0, Y_0)^2\right] - \E\left[Z_0^2\right] \label{eq:varfinal}
\end{align}

Recall again that $Z_0= \Esub{Y_0}\left[f^+(X_0, Y_0) \mid MRX_0Y_1\right]$ is a deterministic function of $MRX_0Y_1$. Hence, we can instead write $\E\left[Z_0^2\right] = \sumsub{mrx_0y_1} \pi(mrx_0y_1) \cdot \left[\Esub{Y_0}\left[f(x_0,Y_0) \mid mrx_0y_1\right]\right]^2.$
Plugging this back into \Cref{eq:varfinal}, we achieve the claim.
\end{proof}

As a result, we obtain a linear decomposition of observer's variances.
\begin{claim} $\VarObs(\pi_0 @ f^{+n/2}) +\VarObs(\pi_1 @ f^{+n/2}) \leq \VarObs(\pi @ f^{+n})$
\label{clm:var_decomp}
\end{claim}

\begin{proof} Following \Cref{clm:var_pi0_pi1}, we derive:
\begin{align*}
    & \VarObs(\pi_0 @ f^{+n/2}) + \VarObs(\pi_1 @ f^{+n/2}) \\
    & =  \Esub{MRX_0Y_1} \ \Varsub{Y_0}\left[f^+(X_0,Y_0) \mid MRX_0Y_1\right] + \Esub{MRX_0Y_1} \ \Varsub{X_1}\left[f^+(X_1,Y_1) \mid MRX_0Y_1\right] \\
    & = \Esub{MRX_0Y_1} \left[\Varsub{Y_0}\left[f^+(X_0,Y_0) \mid MRX_0Y_1\right] + \Varsub{X_1}\left[f^+(X_1,Y_1) \mid MRX_0Y_1\right]\right] 
\end{align*}
Using the rectangle property (\Cref{prop:rectangle}), we have $X_1 \perp Y_0 \mid MRX_0Y_1$, therefore $f^+(X_0,Y_0) \perp f^+(X_1,Y_1) \mid MRX_0Y_1$. We then finish the derivation:
    \begin{align*}
    & = \Esub{MRX_0Y_1} \ \Varsub{X_1Y_0}\left[f^+(X,Y) \mid MRX_0Y_1\right] \tag{$f^+(X_0,Y_0) \perp f^+(X_1,Y_1) \mid MRX_0Y_1$} \\
    & \leq \Esub{MR} \ \Varsub{XY} \left[f^+(X,Y) \mid MR\right] \tag{\Cref{fact:cond_var}} \\
    & = \VarObs(\pi @ f^{+n})
\end{align*}
as desired.
\end{proof}

\paragraph{Lifting the assumptions about $Z_0, Z_1$.} Recall that in the protocol $\pi_0$, we require Alice to compute $Z_0 = \Esub{Y_0}\left[f^+(X_0, Y_0) \mid MRX_0Y_1\right]$ and communicate it \emph{exact} value to Bob. To lift the assumptions, we now instead let Alice sends the value of $Z_0^2$ rounded down to the nearest multiple of $\frac{\Delta}{32}$. We denote such value by a random variable $\hat{Z}_0^2$. Bob also replaces $Z_1$ with $\hat{Z}_1^2$, defined analogously. The implications of such rounding scheme are as follows.

\begin{itemize}
    \item Observe that the range of $Z_0^2$ and $Z_1^2$ is $[0,n^2]$. This allows us to set a discrete space $\mathcal{Z}$ to be $\{0, \frac{\Delta}{32}, \ldots \frac{\Delta}{32}, \ldots, \frac{\Delta}{32} \cdot \lfloor n^2 \cdot \frac{32}{\Delta}\rfloor\}$. As a result, $|\mathcal{Z}| = O(\frac{n^2}{\Delta})$. Plugging this into \Cref{clm:IC_decomp}, the decomposition of information costs may incur an additive loss of at most $O(\log{\frac{n}{\Delta}})$ bits.
    \item In the variance analysis of \Cref{clm:var_pi0_pi1}, the exact value of $Z_0^2$ 
    (and symmetrically, $Z_1^2$) was only used in \Cref{eq:varfinal}. 
    Replacing $Z_0^2$ by $\hat{Z}_0^2$ introduces at most an additive error of $\frac{\Delta}{32}$ 
    to $\VarObs(\pi_0 @ f^{+n/2})$. 
    The same argument applies to $\pi_1$, giving a total additive loss of at most $\frac{\Delta}{16}$ 
    to the combined variance decomposition.
\end{itemize}

In summary, replacing the exact values $Z_0$ and $Z_1$ by their squared and rounded counterparts $\hat{Z}_0^2$ and $\hat{Z}_1^2$ preserves the linear decomposition properties, at the cost of an additive $O(\log \tfrac{n}{\Delta})$ in information and $\tfrac{\Delta}{16}$ in variance. This yields the following decomposition lemma.

\begin{lemma}[One-step decomposition lemma] Let $f: \mathcal{X} \times \mathcal{Y} \rightarrow \{\pm 1\}$ be an arbitrary function, and let $\mu$ be any input distribution over $\mathcal{X} \times \mathcal{Y}$. Suppose $\pi$ is a protocol over the input distribution $\mu^n$. Then, for any $\Delta \in (0,1)$, there exist protocols $\pi_0$ and $\pi_1$ over input distribution $\mu^{n/2}$ such that 
\begin{align*}
    \IC(\pi_0) + \IC(\pi_1) \leq \IC(\pi) + O(\log{\frac{n}{\Delta}}) \hspace{2mm} \text{ and } \hspace{2mm} \VarObs(\pi_0@f^{+n/2}) + \VarObs(\pi_1@f^{+n/2}) \leq \VarObs(\pi@f^{+n}) + \frac{\Delta}{16}.
\end{align*}
\end{lemma}

We are now ready to prove the main result of this section.

\mainlemma*

\begin{proof}
Starting from $\pi$, we perform the decomposition recursively until we reach the leaf level. More formally, we may assume that $n$ is a power of two. Let the starting protocol be $\pi_\emptyset := \pi$ over the input distribution $\mu^n$. We first decompose $\pi$ into $\pi_0$ and $\pi_1$, each over the input distribution $\mu^{n/2}$. We then further decompose $\pi_0$ into $\pi_{00}$ and $\pi_{01}$, and $\pi_1$ into $\pi_{10}$ and $\pi_{11}$, each over the input distribution $\mu^{n/4}$. This process continues recursively until the leaf level, yielding $n$ protocols $\{\pi_S\}_{S \in \{0,1\}^{\log_2{n}}}$, each operates over the input distribution $\mu$.  

Consider a level $\ell \in [\log_2{n}]$ of the decomposition where we decompose $2^{\ell-1}$ protocols $\{\pi_S\}_{S \in \{0,1\}^{\ell-1}}$ into $2^{\ell}$ protocols $\{\pi_{S}\}_{S \in \{0,1\}^{\ell}}$. There are $2^{\ell-1}$ decompositions at this level, each introducing an additive loss of $O\!\left(\log{\frac{n \cdot 2^{-\ell}}{\Delta}}\right)$ bits for information and an additive loss of $\frac{\Delta}{16}$ for variances. Summing over all levels gives $$\sum_{S \in \{0,1\}^{\log_2 n}} \VarObs(\pi_S@f) \leq \VarObs(\pi@f^{+n}) + \sum_{\ell=1}^{\log_2 n} 2^{\ell-1} \cdot \frac{\Delta}{16} \leq \VarObs(\pi@f^{+n}) + n\Delta,$$
and
$$\sum_{S \in \{0,1\}^{\log_2 n}} \IC(\pi_S) \leq \IC(\pi) + \sum_{\ell=1}^{\log_2 n} 2^{\ell-1} \cdot O\!\left(\log{\frac{n \cdot 2^{-\ell}}{\Delta}}\right) \leq \IC(\pi) + O\!\left(n + n \log{\frac{1}{\Delta}}\right).$$

This means an average protocol $\eta := \pi_S$ where $S$ is a uniformly random string in $\{0,1\}^{\log_2{n}}$ has $\IC(\eta) \leq \frac{\IC(\pi)}{n} + O(1+\log\frac{1}{\Delta})$ and $\VarObs(\eta@f) \leq \frac{\VarObs(\pi@f^{+n})}{n} + \Delta$. This concludes the proof.
\end{proof}

\section{Round-Preserving Direct-Sum for Variance and Information}
\label{sec:direct_sum_oneway}

We now prove \Cref{lem:main_lemma_oneway}, thereby completing the proof of the round-preserving gap-majority lemma (\Cref{thm:round_gap_majority}). We adopt the setting of the lemma: let $f: \mathcal{X} \times \mathcal{Y} \rightarrow \{\pm 1\}$ be an arbitrary function, and let $\mu$ be an arbitrary input distribution. Denote by $\pi = (X,Y,M,A,R)$ an $r$-round protocol over the product distribution $\mu^n$, where $X = (X_1,\ldots,X_n)$ and $Y = (Y_1,\ldots,Y_n)$ with each $(X_i,Y_i)$ drawn independently from $\mu$. We also write $M = (M_1,\ldots,M_r)$ as a sequence of $r$ messages. We may assume $r$ is odd so that Alice sends the last message $M_r$ and Bob is responsible for outputting the answer $A$ from the answer space $\mathcal{A}$. Note that Alice may not know this answer $A$. Let $\Delta \in (0,1)$ be an arbitrary additive error parameter. 

In a nutshell, the exact decomposition procedure from \Cref{sec:direct_sum_var} fails here for two reasons. First, in the protocol $\pi_1$ (\Cref{pi1}), Bob must send $Z_1$ back to Alice. This takes an extra round of communication to which we cannot afford. To overcome this, we change the structure of our decomposition. Specifically, we instead decompose the protocol $\pi$ operating on $n$ input coordinates into two subprotocols: $\pi^{<n}$ which operates on the first $n-1$ input pairs, and $\pi^{n}$ which operates on the single final input pair. Additionally, after simulating $\pi$ in $\pi^{<n}$ (similar to Line 5 in \Cref{pi0}), Alice will send an extra auxiliary message to Bob simultaneously with her last message $M_r$ (note that this does \emph{not} cost an extra round), whereas in $\pi^{n}$, Bob will \emph{no longer} send any extra message back to Alice! This modification preserves the number of rounds of both subprotocols at $r$ while introducing only a larger-yet-manageable additive loss in the information cost.

Second, in $r$-round protocols (when $r$ is odd), only Bob is required to produce the output. Consequently, the observer’s variance is no longer an accurate proxy for the error, since Bob cannot use an additional round to relay the answer to Alice, preventing the external observer from seeing it. To address this, we use the reformulation of variances that is tailored to $r$-round protocols (recall \Cref{def:1-way-var}). Precisely, we show that the decomposition of variances similar to \Cref{clm:decomp_IC_1way} still hold, if we instead measure the variance of $\pi^n$ from Bob's perspective.

\paragraph{Decomposition.} We now describe the modified decomposition scheme. Given an $r$-round protocol $\pi$, the decomposition produces two smaller $r$-round subprotocols $\pi^{<n}$ and $\pi^n$ described in \Cref{pi<n} and \ref{pi_n}. Note that $\pi^{<n}$ operates on an input distribution $(X',Y') \sim \mu^{n-1}$ and $\pi^{n}$ operates on an input distribution $(X',Y') \sim \mu$.

\begin{figure}[H]
    \centering
    \begin{tabular}{|p{15.7cm}|}
    \hline ~\\
    \multicolumn{1}{|c|}{\textbf{A protocol $\pi^{<n}$ over the input distribution $(X',Y') \sim \mu^{n-1}$}} \\ 
    \\ \hline
    \vspace{-2.85mm}

    % \textbf{Procedures:}
    \begin{enumerate}
         \item Alice pretends that her input consists of $n$ coordinates, and writes it as $X = (X_1,\ldots,X_n)$.
         \vspace{-5.8mm}
        \item Bob pretends that his input consists of $n$ coordinates, and writes it as $Y = (Y_1,\ldots,Y_n)$.
        \vspace{-1mm}
        \item Alice embeds $X'$ into $X_{<n}$, and Bob embeds $Y'$ into $Y_{<n}$.
        \vspace{-1mm}
        \item Players use public randomness to draw $Y_n$. Alice then privately draws $X_n$.
        \vspace{-1mm}
        \item Alice and Bob simulate the $r$ rounds of communication $M = (M_1,\ldots,M_r)$ with respect to the protocol $\pi$ and the input $(X,Y)$.
        \vspace{-1mm}
        \item For each $a \in \A$, Alice computes $Z_n^{a} := \Esub{Y_{<n}} \left[f^+(X_{<n}, Y_{<n}) \mid MRX_{<n}Y_n, A=a\right]$ and then sends $Z_n := \{Z_n^a\}_{a \in \A}$ to Bob along with $M_r$.
        \vspace{-1mm}
        \item Bob generates $A$.
    \end{enumerate}
    \vspace{0mm}
    \\
    \hline
    \end{tabular}
    \caption{The protocol $\pi^{<n}$ resulting from decomposing $\pi$}
    \label{pi<n}
\end{figure}

\begin{figure}[H]
    \centering
    \begin{tabular}{|p{15.7cm}|}
    \hline ~\\
    \multicolumn{1}{|c|}{\textbf{A protocol $\pi^{n}$ over the input distribution $(X',Y') \sim \mu$}} \\ 
    \\ \hline
    \vspace{-2.85mm}
    % \textbf{Procedures:}
    \begin{enumerate}
        \item Alice pretends that her input consists of $n$ coordinates, and writes it as $X = (X_1,\ldots,X_n)$. 
        \hspace{-1mm}
        \item Bob pretends that his input consists of $n$ coordinates, and writes it as $Y = (Y_1,\ldots,Y_n)$.
        \hspace{-1mm}
        \item Alice embeds $X'$ into $X_n$, and Bob embeds $Y'$ into $Y_n$.
        \hspace{-1mm}
        \item Players use public randomness to jointly draw $X_{<n}$. Bob then privately draws $Y_{<n}$. 
        \hspace{-1mm} 
        \item Alice and Bob simulate the $r$ rounds of communication $M = (M_1,\ldots,M_r)$ with respect to the protocol $\pi$ and the input $(X,Y)$.
        \vspace{-1mm}
        \item Bob generates $A$.
    \end{enumerate}
    \vspace{0mm}
    \\
    \hline
    \end{tabular}
    \caption{The protocol $\pi^n$ resulting from decomposing $\pi$.}
    \label{pi_n}
\end{figure}

Equivalently, we may adopt the following distributional view of each protocol.  
\begin{itemize}
    \item $\pi^n = (X_n, Y_n, M, A, RX_{<n})$.
    \item $\pi^{<n} = (X_{<n}, Y_{<n}, M \odot Z_n, A, RY_{n})$ where $Z_n$ consists of $Z_n^a = \Esub{Y_{<n}} \left[f^+(X_{<n}, Y_{<n}) \mid MRX_{<n}Y_n, A=a\right]$ for each $a \in \A$ computable solely by Alice. Here we write $M \odot Z_n$ to denote the $r$-round transcript $M$ with Alice appending $Z_n$ to her last message $M_r$. Note further that for a fixed value of $a\in \A$, $Z^a_n$ is a deterministic function of $MRX_0Y_1$.
\end{itemize}

For convenience, we make the same assumption as in the previous section: each $Z_n^{a}$ takes values from a discrete set $\mathcal{Z}$. As before, we will lift this assumption shortly.

\paragraph{Decomposition of Information Costs.} Analogous to the proof of \Cref{clm:IC_decomp}, without the $Z_n$ in $\pi^{<n}$, the information costs of $\pi^n$ and $\pi^{<n}$ sums up to the information cost of $\pi$. The extra message $Z_n$ of $\pi^{<n}$ incurs additional $O(|\A| \cdot \log{|\mathcal{Z}|})$ bits, as there are $|\A|$ choices of $a$, and each $Z_n^{a}$ can be represented with $O(\log{|\mathcal{Z}|})$ bits. This yields the following information cost decomposition. We omit its formal proof for brevity.
\begin{claim} $\IC(\pi^{n}) + \IC(\pi^{<n}) \leq \IC(\pi) + O(|\A| \cdot \log{|\mathcal{Z}|})$.
\label{clm:decomp_IC_1way}
\end{claim}

\paragraph{Decomposition of Variances.} Analogous to \Cref{clm:decomp_var_oneway}, we show that the decomposition of variance still holds, if the variance of $\pi^n$ is now evaluated from Bob's perspective.

\begin{claim} $\VarBobRound(\pi^n @ f) + \VarObsRound(\pi^{<n} @ f^{+(n-1)}) \leq \VarObsRound(\pi @ f^{+n})$.
\label{clm:decomp_var_oneway}
\end{claim}
\begin{proof}
By definition, with $\pi^n = (X_n, Y_n, M, A, RX_{<n})$, we have 
\begin{equation}
    \VarBobRound(\pi^n @ f) = \Esub{MRAX_{<n}Y_n} \ \Varsub{X_n} \left[f(X_n,Y_n) \mid MRAX_{<n}Y_n\right]
    \label{eq:gamma_Bob_pi_n}
\end{equation}
and with and $\pi^{<n} = (X_{<n}, Y_{<n}, M \odot Z_n, A, RY_{n})$, we have
\begin{align}
    \VarObsRound(\pi^{<n} @ f^{+(n-1)}) & = \Esub{MRAY_nZ_n} \ \Varsub{X_{<n}Y_{<n}} \left[f^+(X_{<n},Y_{<n}) \mid MRAY_nZ_n\right] \nonumber \\
    & = \sum_{a \in \A} \pi(a) \cdot \Esub{MY_nZ_n} \ \Varsub{X_{<n}Y_{<n}} \left[f^+(X_{<n},Y_{<n}) \mid MRY_nZ_na\right] \nonumber\\
    & \leq \sum_{a \in \A} \pi(a) \cdot \Esub{MY_nZ^a_n} \ \Varsub{X_{<n}Y_{<n}} \left[f^+(X_{<n},Y_{<n}) \mid MRY_nZ^a_na\right] \label{eq:gamma_<n_obs_1}
\end{align}
where the final inequality follows \Cref{fact:cond_var} with $Z^a_n \in Z_n$. Now fix $a \in \A$ and expand the expectation.
\begin{align}
    & \Esub{MRY_nZ^a_n} \ \Varsub{X_{<n}Y_{<n}} \left[f^+(X_{<n},Y_{<n}) \mid MRY_nZ^a_na\right] \nonumber \\
    & =  \Esub{X_{<n}Y_{<n}} \left[f^+(X_{<n}, Y_{<n})^2 \mid a\right] - \sum_{mry_nz_n^a} \pi(mry_nz_n^a\mid a) \cdot \left[\Esub{X_{<n}Y_{<n}} \left[f(X_{<n},Y_{<n}) \mid mry_nz_n^a a\right]\right]^2 \label{eq:gamma_<n_obs_2}
\end{align}
We further unpack the final expectation term.
\begin{align*}
    \Esub{X_{<n}Y_{<n}} \left[f(X_{<n},Y_{<n}) \mid mry_nz_n^a a\right] & = \sum_{x_{<n}} \pi(x_{<n} \mid mry_nz_n^a a) \cdot \Esub{Y_{<n}} \left[f^+(x_{<n}, Y_{<n}) \mid mrx_{<n}y_nz_n^a a\right] \\
    & = \sum_{x_{<n}} \pi(x_{<n} \mid mry_nz_n^a a) \cdot \Esub{Y_{<n}} \left[f^+(x_{<n}, Y_{<n}) \mid mrx_{<n}y_n a\right] \tag{$z_n^a$ is deterministic given $mrx_{<n}y_na$} \\
    & = \sum_{x_{<n}} \pi(x_{<n} \mid mry_nz_n^a a) \cdot z_n^a \tag{definition of $z_n^a$} \\
    & = z_n^a.
\end{align*}
Plugging this into \Cref{eq:gamma_<n_obs_2}, we obtain:
\begin{align}
    \Esub{MRY_nZ^a_n} \ \Varsub{X_{<n}Y_{<n}} \left[f^+(X_{<n},Y_{<n}) \mid MRY_nZ^a_na\right] & =  \Esub{X_{<n}Y_{<n}} \left[f^+(X_{<n}, Y_{<n})^2 \mid a\right] - \sum_{mry_nz_n^a} \pi(mry_nz_n^a\mid a) \cdot (z^a_n)^2 \nonumber \\
    & =  \Esub{X_{<n}Y_{<n}} \left[f^+(X_{<n}, Y_{<n})^2 \mid a\right] - \Esub{Z_n^a} \left[(Z^a_n)^2 \mid a\right] \label{eq:gamma_<n_obs_3}
\end{align}
Recall again that for a fixed $a \in \A$, $Z_n^{a} = \Esub{Y_{<n}} \left[f^+(X_{<n}, Y_{<n}) \mid MRX_{<n}Y_n, A=a\right]$ is a deterministic function of $MRX_{<n}Y_n$. Thus, we can write :
\begin{align}
\Esub{Z_n^a} \left[(Z^a_n)^2 \mid a\right] = \sum_{mrx_{<n}y_n} \pi(mrx_{<n}y_n \mid a) \cdot \left[\Esub{Y_{<n}} \left[f^+(X_{<n}, Y_{<n}) \mid mrx_{<n}y_na\right]\right]^2     \label{eq:gamma_<n_obs_4}
\end{align}
Combining \Cref{eq:gamma_<n_obs_1}, (\ref{eq:gamma_<n_obs_3}), and (\ref{eq:gamma_<n_obs_4}), we derive:
\begin{align}
    & \VarObsRound(\pi^{<n} @ f^{+(n-1)}) \nonumber \\ 
    & \leq \sum_{a \in \A} \pi(a) \cdot \Esub{MRY_nZ^a_n} \ \Varsub{X_{<n}Y_{<n}} \left[f^+(X_{<n},Y_{<n}) \mid MY_nZ^a_na\right] \nonumber \\
    & = \Esub{X_{<n}Y_{<n}}\left[f^+(X_{<n}, Y_{<n})^2\right] - \sum_{mrx_{<n}y_na} \pi(mrx_{<n}y_n a) \cdot \left[\Esub{Y_{<n}} \left[f^+(X_{<n}, Y_{<n}) \mid mrx_{<n}y_na\right]\right]^2 \nonumber \\
    & = \Esub{MRX_{<n}Y_nA} \ \Varsub{Y_{<n}} \left[f^+(X_{<n}Y_{<n}) \mid MRX_{<n}Y_nA\right] \label{eq:gamma_<n_obs_final}
\end{align}

Finally, combining \Cref{eq:gamma_Bob_pi_n} and (\ref{eq:gamma_<n_obs_final}), we obtain:
\begin{align*}
     & \VarBobRound(\pi^n @ f) + \VarObsRound(\pi^{<n} @ f^{+(n-1)}) \\
     & = \left[\Esub{MRAX_{<n}Y_n} \ \Varsub{X_n} \left[f(X_n,Y_n) \mid MRAX_{<n}Y_n\right]\right] + \left[ \Esub{MRAX_{<n}Y_n} \ \Varsub{Y_{<n}} \left[f^+(X_{<n},Y_{<n}) \mid MRAX_{<n}Y_n\right]\right] \\
     & = \Esub{MRAX_{<n}Y_n} \ \Varsub{X_nY_{<n}} \left[f^+(X,Y) \mid MRAX_{<n}Y_n\right] \tag{rectangle property} \\
     & \leq \Esub{MRA} \ \Varsub{XY} \left[f^+(X,Y) \mid MRA\right] \tag{\Cref{fact:cond_var}} \\
     & = \VarObsRound(\pi @ f^{+n}) \tag{definition}
\end{align*}
as desired.
\end{proof}

\paragraph{Lifting the assumptions about $Z_n$.} We now lift the assumptions of exact communication of $Z_n$ in the same fashion as we did in \Cref{sec:direct_sum_var}. For each $a \in \A$, instead of Alice sending each $Z_n^a$ exactly, she sends the value of $(Z_n^a)^2$ rounded down to closest multiple of $\frac{\Delta}{32}$.  The implications of such rounding scheme are as follows.
\begin{itemize}
    \item Observe that the range of each $Z_n^{a}$ is $[0,n^2]$. This allows us to set a discrete space $\mathcal{Z}$ to be $\{0, \frac{\Delta}{32}, \ldots \frac{\Delta}{32}, \ldots, \frac{\Delta}{32} \cdot \lfloor n^2 \cdot \frac{32}{\Delta}\rfloor\}$. As a result, $|\mathcal{Z}| = O(\frac{n^2}{\Delta})$. Plugging this into \Cref{clm:decomp_IC_1way}, the decomposition of information costs may incur an additive loss of $O(|\mathcal{A}| \cdot \log{\frac{n}{\Delta}})$ bits.
    \item Similar to before, such rounding scheme introduces an additive loss to observer variance of $\pi^{<n}$ by $\frac{\Delta}{32}$. Here, though, Bob's variace does not suffer any losses. As a result, the decomposition of variances may incur an additive loss of $\frac{\Delta}{32}$.
\end{itemize}
In summary, by replacing the exact communication of $Z_n$ by its squared and rounded version, we obtain the following round-preserving decomposition lemma.

\begin{lemma}[One-step round-preserving decomposition lemma] Let $f: \mathcal{X} \times \mathcal{Y} \rightarrow \{\pm 1\}$ be an arbitrary function and let $\mu$ be an arbitrary input distribution over $\mathcal{X} \times \mathcal{Y}$. Suppose $\pi$ is an $r$-round protocol over input distribution $\mu^n$ which Bob produces an answer from a discrete space $\A$. For any $\Delta \in (0,1)$, there exists $r$-round protocols $\pi^n$ over input distribution $\mu$ and $\pi^{<n}$ over input distribution $\mu^{n-1}$ such that $$\IC(\pi^n) + \IC(\pi^{<n}) \leq \IC(\pi) + O(|\A| \cdot \log{\frac{n}{\Delta}})$$ and $$\VarBobRound(\pi^n@f) + \VarObsRound(\pi^{<n}@f^{+(n-1)}) \leq \VarObsRound(\pi@f^{+n}) + \frac{\Delta}{32}.$$
\label{lem:one_step_one_way}
\end{lemma}

We are now ready to prove the main result of this section.

\mainlemmaoneway*

\begin{proof}
    Consider the following recursive decomposition procedure. We first decompose $\pi$ into $\pi^{n}$ and $\pi^{<n}$. Next, we decompose $\pi^{<n}$ into $\pi^{n-1}$ and $\pi^{<n-1}$. This process is repeated $n-1$ times, until we obtain $n$ protocols $\{\pi^i\}_{i \in [n]}$, each operating over a single input pair drawn from $\mu$. Note further that every decomposition preserves the number of rounds. Thus, each protocol has $r$ rounds of communication.

Applying \Cref{lem:one_step_one_way} recursively, we obtain
$$\sum_{i \in [n]} \VarBobRound(\pi^{i} @ f) \leq \VarObsRound(\pi @ f^{+n}) + \frac{n \Delta}{16} \hspace{3mm} \text{ and } \hspace{3mm} \sum_{i \in [n]} \IC(\pi^i) \leq \IC(\pi) + O(|\A| \cdot n \log{\frac{n}{\Delta}}).$$

This means an average protocol $\eta := \pi^i$, where $i$ is chosen uniformly at random from $[n]$, satisfies $\IC(\eta) \leq \frac{\IC(\pi)}{n} + O(|\A| \cdot \log{\frac{n}{\Delta}})$ and $\VarBobRound(\eta@f) \leq \frac{\VarObsRound(\pi@f^{+n})}{n} + \Delta$. This concludes the proof.
\end{proof}

\section{Lower Bound Applications}
\label{sec:lb_apps}

Finally, we demonstrate the power of the gap-majority lemma by using it to derive lower bounds for the communication complexity of Gap-Hamming and the streaming complexity of triangle counting. We begin by revisiting a few fundamental communication problems and their \emph{information} lower bounds. These problems, throughout this section, will serve as the base functions for our gap-majority reduction. The first problem is \emph{set-disjointness}.
\begin{definition}[Set-Disjointness]
Let $k$ be a positive integer. The communication problem \emph{(unique) set-disjointness}, denoted $\DISJ_k$, is defined as follows. Alice receives a set $X \subseteq [k]$ and Bob receives a set $Y \subseteq [k]$ such that either $|X \cap Y| = 0 $ or $|X \cap Y| = 1$. They must decide whether their sets are disjoint, i.e. they wish to compute
$$
\DISJ_k(X,Y) := 
\begin{cases} 
-1 & \text{ if } |X \cap Y| = 0, \\
+1 & \text{ if } |X \cap Y| = 1.
\end{cases}
$$
Alternatively, we may write $|X \cap Y| = \frac{1}{2} \cdot \left(\DISJ_k(X,Y) + 1\right).$
\end{definition}

\begin{theorem} There is an input distribution $\mathcal{D}_{\DISJ_k}$ such that any protocol $\pi$ that computes $\DISJ_k$ with probability $0.501$ over the distribution $\mathcal{D}_{\DISJ_k}$ requires $\IC(\pi) \ge \Omega(k)$. In addition, the distribution $\mathcal{D}_{\DISJ_k}$ is balanced, meaning $\E\left[{\DISJ_k}(X,Y)\right] = 0$ under the randomness of $\mathcal{D}_{\DISJ_k}$.
\label{thm:disj_info_lb}
\end{theorem}

The next problem is \emph{indexing}. In fact, we will use a slightly non-standard variant of indexing which restricts the support size.

\begin{definition}[Indexing]
Let $N,k$ be positive integers such that $k \le \frac{N}{2}$. The \emph{one-way} communication problem \emph{indexing}, denoted $\IND_{N,k}$, is defined as follows. Alice receives an $N$-digit string $S \in \{\pm 1\}^N$ which contains exactly $k$ $1$'s, and Bob receives an index $\ell \in [N]$. Alice sends a single message to Bob, who then wishes to compute the $\ell^{\text{th}}$ digit of $S$, namely $
\IND_{N,k}(S,\ell) := S_\ell.$
\end{definition}

\begin{theorem}
There is an input distribution $\mathcal{D}_{\IND_{N,k}}$ such that any protocol $\pi$ that computes $\IND_{N,k}$ with probability $0.501$ over the distribution $\mathcal{D}_{\IND_{N,k}}$ requires $\IC(\pi) \ge \Omega(k)$. In addition, the distribution $\mathcal{D}_{\IND_{N,k}}$ is balanced, meaning $\E\left[{\IND_{N,k}}(X,Y)\right] = 0$ under the randomness of $\mathcal{D}_{\IND_{N,k}}$.
\label{thm:ind_info_lb}
\end{theorem}

We note that there are many known information-theoretic proofs of these lower bounds; see, for example, the folklore arguments of \cite{KalyanasundaramS92,Razborov92} or the comprehensive textbook by Rao and Yehudayoff \cite{Rao_Yehudayoff_2020}.

\subsection{Gap-Hamming}
\label{subsec:GH}

Consider the following \emph{Gap-Hamming} problem in the two-player communication setting.

\begin{definition}[Gap-Hamming] Let $n$ be a positive integer. The communication problem \emph{Gap-Hamming}, denoted $\GH_n$, is defined as follows: Alice and Bob receives $X,Y \in \{\pm 1\}^n$ respectively, and they wish to distinguish whether the inner product of their input vectors is more than $0.01 \sqrt{n}$ or less than $-0.01\sqrt{n}$. Formally, they wish to compute
    $$\GH_n(X,Y) := \begin{cases} +1 & \text{ if $\langle X, Y \rangle \geq 0.01\sqrt{n}$} \\  -1 & \text{ if $\langle X, Y \rangle \leq -0.01\sqrt{n}$}
    \end{cases}$$
\end{definition}

Since its formal introduction by Indyk and Woodruff \cite{IndykW03}, Gap-Hamming has quickly become one of the central problems in communication complexity. Due to its ``gap'' structure, it has served as a key tool for proving lower bounds for a variety of problems, such as streaming frequency moment estimation \cite{AlonMS99,Woodruff04,IndykW05,JayramW13}, sketching complexity of graph cuts \cite{AndoniCKQWZ16}, testing of boolean function properties \cite{BlaisBM12}, to name just a few. Notably, early applications of the Gap-Hamming problem were restricted to the \emph{one-way} communication in which the optimal lower bound $\Omega(n)$ was quickly established in the original paper of \cite{IndykW03}. It was only in 2012 that Chakrabarti and Regev~\cite{ChakrabartiR12} proved that $\Omega(n)$ bits of communication are required even for protocols with unbounded rounds. Though their original proof was quite technically involved, a sequence of follow-up works \cite{Vidick12, Sherstov12} provided significantly simpler proofs. Another set of milestones was achieved by \cite{ChakrabartiKW12, KerenidisLLRX15, BravermanGPW16}, who obtained the explicit \emph{information} lower bound: they showed that, under various distributions, any protocol computing $\GH_n$ must reveal $\Omega(n)$ bits of information about the players' inputs.

As briefly discussed earlier, $\GH_n$ is in fact a special case of the gap-majority composition obtained by taking $f(x,y) := xy$. Therefore, it is almost mandatory for us to recover its optimal $\Omega(n)$ lower bound through our gap-majority framework. We remark although the lower bound result is not quantitatively novel, we consider our proof strategy to be the main takeaway message: utilizing the gap-majority lemma as a ``black-box'' tool, our lower bound proof of $\GH_n$ is merely a simple and clean reduction from the classical set-disjointness problem.

\begin{theorem}
Any randomized protocol that computes $\GH_n$ with probability $0.99$ requires $\Omega(n)$ bits of communication.
\label{thm:GH_lb}
\end{theorem}

The rest of this subsection will prove \Cref{thm:GH_lb}. We first describe its hard distribution in which we will prove an \emph{information} lower bound against.

\paragraph{Hard distribution $\mathcal{D}_{\GH}$ for $\GH_n$.}
Let $n = 4k^{2}t$ for a sufficiently large constant $k$. We partition the $n$ coordinates into $t$ disjoint blocks of size $4k^{2}$. For each block $i \in [t]$, Alice receives $X_i \in \{\pm 1\}^{4k^{2}}$ and Bob receives $Y_i \in \{\pm 1\}^{4k^{2}}$, jointly drawn from the following distribution $\mathcal{D}_{\text{block}}$.

% \footnote{$\mathcal{D}_{\GH}$ is exactly the product distribution $\mathcal{D}_{\text{block}}^{t}$.}

\begin{enumerate}
    \item Sample $(X'_i, Y'_i) \sim \mathcal{D}_{\DISJ_k}$.
    
    \item The first $3k^{2}$ coordinates of $X_i$ and $Y_i$ are indexed by $[k] \times [3k]$ and constructed as follows.  
    For each $l \in [k]$:
    \begin{itemize}
        \item If $l \in X'_{i}$, set $X_i(l,j) = -1$ for $1 \le j \le k$, and $+1$ otherwise.
        \item If $l \notin X'_{i}$, set $X_i(l,j) = -1$ for $k+1 \le j \le 2k$, and $+1$ otherwise.
        \item If $l \in Y'_{i}$, set $Y_i(l,j) = -1$ for $1 \le j \le k$, and $+1$ otherwise.
        \item If $l \notin Y'_{i}$, set $Y_i(l,j) = -1$ for $2k+1 \le j \le 3k$, and $+1$ otherwise.
    \end{itemize}
    
    \item The next $k^{2}-k$ coordinates of both $X_i$ and $Y_i$ are set to $+1$.
    
    \item The final $k$ coordinates of $X_i$ are set to $+1$, and the final $k$ coordinates of $Y_i$ are set to $-1$.
\end{enumerate}

We now prove the linear lower bound of gap-hamming.

\begin{proof}[Proof of \Cref{thm:GH_lb}.]
Our strategy is to prove an $\Omega(n)$ \emph{information} lower bound against the distribution $\mathcal{D}_{\GH}$, which directly implies \Cref{thm:GH_lb}.  Let $\pi$ be a protocol for $\GH_n$ that succeeds with probability $0.99$.  When run on inputs drawn from $\mathcal{D}_{\GH} \equiv \mathcal{D}_{\text{block}}^{t}$, the protocol still succeeds with probability $0.99$.

Observe that by the construction of $\mathcal{D}_{\text{block}}$, we have $\langle X_i, Y_i \rangle = 2k$ if $|X'_i \cap Y'_i| = 1$, and $\langle X_i, Y_i \rangle = -2k$ if $|X'_i \cap Y'_i| = 0$. In other words, we can write $\langle X_i, Y_i \rangle = 2k \cdot \DISJ_k(X'_i, Y'_i)$, and therefore $\langle X, Y \rangle 
    = \sum_{i=1}^{t} \langle X_i, Y_i \rangle
    = 2k \cdot \sum_{i=1}^{t} \DISJ_k(X'_i, Y'_i),$
where the terms $\DISJ_k(X'_i, Y'_i)$ are drawn i.i.d. from $\mathcal{D}_{\DISJ_k}$ and each takes values in $\{\pm 1\}$ with $0$ mean.

Since $\pi$ distinguishes $\langle X, Y \rangle$ between $\geq 0.01\sqrt{n}$ and $\leq -0.01\sqrt{n}$ with probability 0.99, it also distinguishes $\sum_{i \in [t]} \DISJ_k(X'_i, Y'_i)$ from $\geq \frac{0.01\sqrt{n}}{2k} = 0.01\sqrt{t}$ and $\leq -\frac{0.01\sqrt{n}}{2k} = -0.01\sqrt{t}$ with probability 0.99. Thus, the premises of our gap-majority lemma (\Cref{thm:gap_majority}) are met. 

Applying the gap-majority lemma, we obtain a protocol $\eta$ for computing $\DISJ_k$ over $\mathcal{D}_{\DISJ_k}$ with sucesss probability $0.501$ and $\IC(\eta) \leq \frac{1}{t} \cdot \IC(\pi) + O(1)$. By \Cref{thm:disj_info_lb}, this protocol $\eta$ must have $\IC(\eta) \geq \Omega(k)$. Combining these inequalities, we obtain $\IC(\pi) \geq t \cdot \left(\Omega(k) - O(1)\right)$ which implies $\IC(\pi) \geq \Omega(kt) = \Omega(n)$ using that $k$ is a sufficiently large constant and $kt = \frac{n}{4k} = \Theta(n)$.  
\end{proof}

\subsection{Triangle Counting}
\label{subsec:tricount}

Consider the following \emph{triangle counting} problem in a single-pass streaming setting.

\begin{definition}[Streaming triangle counting] The streaming problem \emph{triangle counting}, denoted $\tricount$, is defined as follows. An $n$-vertex $m$-edge graph $G = (V,E)$ is given as a stream of edge arrival and let $\varepsilon \in (0,1)$ be a multiplicative error parameter. Let $T$ be the number of triangles in $G$ unknown to us. As a courtesy, we are given $\tau$ which is an $O(1)$-approximation of $T$ at the beginning of the stream. At the end of the stream, we wish compute $\widetilde{T}$ such that $|\widetilde{T}-T| \leq \eps T$.
\end{definition}

A naive streaming algorithm for triangle counting simply stores the entire graph as the stream progresses, which requires $\widetilde{O}(m)$ bits of memory. Beyond this trivial approach, a long line of work has culminated in a randomized algorithm that uses $\widetilde{O}\!\left(\frac{mn}{\varepsilon^2 T}\right)$ space \cite{Bar-YossefKS02, JowhariG05,BuriolFLMS06, BravermanOV13, KallaugherP17, JayaramK21}. In this section, we show that, up to $\operatorname{polylog}(n)$ factors, these two upper bounds taken together are in fact optimal \emph{even with respect to the dependence on $\varepsilon$}.

\begin{theorem} There exist absolute constants $c_1,c_2 >0$ such that the following statement is true. For any $\tau \leq c_1 \cdot \min\{m,\frac{n^3}{m}\}$ and $\eps \in (0,c_2)$, any streaming algorithms that answer $\tricount$ with $\widetilde{T}$ such that $|\widetilde{T}-T| \leq \eps T$ with probability $0.999$ must use $\Omega\left(\min\{m, \frac{mn}{\eps^2T}\}\right)$ bits of space.
\label{thm:tricount_lb}    
\end{theorem}

Our proof of \Cref{thm:tricount_lb} follows a standard approach by reducing $\tricount$ to a new one-way communication problem $\mathcal{P}$ (state shortly). We assume the setting of parameters $m,n,\tau,T$ as stated in \Cref{thm:tricount_lb}, and may assume further that $\eps^2 T \geq n$. Let $\mathcal{A}$ be a streaming algorithm with space $S$ and succeeds for $\tricount$ with probability $0.999$. Our ultimate goal is to show that $S \geq \Omega\left(\frac{mn}{\eps^2T}\right)$. We note that in our upcoming reduction, we shall construct a graph with $\Theta(n)$ vertices, $\Theta(m)$ edges, and $\Theta(T)$ triangles instead of exactly $n$ vertices, $m$ edges and $T$ triangles. These numbers can then be scaled as needed.

Set $t = \Theta(1/\eps^{2})$, $k = \Theta(\sqrt{\frac{mn}{\tau}})$, and $b = \Theta(\eps^2 \tau/n)$ for suitable leading constants. We also ensure that $t,k,b \geq 1$ which is possible due to the set-up of our parameters. Consider the following one-way (distributional) communication problem $\mathcal{P}$ which can be interpreted as a $\pm 0.01\sqrt{t}$ approximation of $\mathsf{IND}_{k^2, k^2/10}^{+t}$ over the distribution $\mathcal{D}_{\mathsf{IND}_{k^2, k^2/10}}^{t}$.
\begin{enumerate}
    \item Sample $t$ i.i.d. instances of $\mathcal{D}_{\mathsf{IND}_{k^2, k^2/10}}$, namely $(S_{i}, \ell_{i})$ for each $i,j \in [t]$.
    \item Alice receives all the $\{S_{i}\}_{i \in [t]}$ and Bob receives all the $ \{\ell_{i}\}_{i \in [t]}$.
    \item Alice sends one message $M$ to Bob, who then are asked to compute $\sum_{i \in [t]} \mathsf{IND}_{k^2, k^2/10}(S_{i}, \ell_{i})$ within $\pm 0.01\sqrt{t}$ error.
\end{enumerate}

We will first prove the lower bound of $\mathcal{P}$ via our round-preserving gap-majority lemma taken with $r=1$.

\begin{claim} Assume the setting of parameters $m,n,T,$ and $\eps$.  Then, any one-way communication protocol that solves $\mathcal{P}$ with probability $0.99$ requires $\Omega\left(\frac{mn}{\eps^2T}\right)$ bits of communication.
\label{clm:pi_lb}
\end{claim}
\begin{proof} Let $\pi$ be such protocol using $C$ bits of communication. Using \Cref{thm:round_gap_majority} with $r=1$, we obtain a one-way protocol $\eta$ for computing $\mathsf{IND}_{k^2, k^2/10}$ over the distribution $\mathcal{D}_{\mathsf{IND}_{k^2, k^2/10}}$ that succeeds with probability $0.501$ and $\IC(\eta) \leq \frac{1}{t} \cdot \IC(\pi) + O(\log{t})$. By \Cref{thm:ind_info_lb}, we have $\IC(\eta) \geq \Omega(k^2) = \Omega\left(\frac{mn}{T}\right)$. Combining all with $t = \Theta(1/\eps^2)$, we have:
$$C \geq \IC(\pi) \geq \frac{1}{\eps^2} \cdot \left(\Omega\left(\frac{mn}{T}\right) - O\left(\log{(1/\eps)}\right)\right) \geq \Omega\left(\frac{mn}{\eps^2 T}\right)$$
where the final inequality uses the fact that $T \leq O(\tau) \leq \Theta(m)$.
\end{proof}

We now make the reduction by constructing a one-way protocol $\pi$ for $\mathcal{P}$ that succeeds with probability $0.99$. To do so, the players use their input to construct a graph $G$ of $\Theta(n)$ vertices, $\Theta(m)$ edges, and $\Theta(T)$ triangles, and then invoke the streaming algorithm $\mathcal{A}$ for $\tricount$. 

\paragraph{Vertices.} $G$ has $n+2tkb = \Theta(n)$ vertices where the first $n$ vertices are labeled $z_{1},...,z_{n}$, and the remaining $2tkb$ vertices are split into $2t$ disjoint blocks of size $kb$ named $X_{1},\ldots,X_{t}$ and $Y_{1},\ldots,Y_{t}$. Each block $X_i$ and $Y_i$ is further split into $b$ mini-blocks of size $k$ named $X_{i,1},\ldots,X_{i,b}$, and $Y_{i,1},\ldots,Y_{i,b}$ respectively. There are no edges within the $X$'s parts nor within the $Y$'s parts.

\paragraph{Edges and edge-partition.}  For each $i \in [t]$, consider the input pair $(S_{i}, \ell_{i}) \in \{\pm 1\}^{k^2} \times [k^2]$ given to Alice and Bob respectively. For each $j \in [b]$, Alice views the induced subgraph $X_{i,j} \cup Y_{i,j}$ as a vector of dimension $k^2$. She then encodes $S_{i} \in \{\pm 1\}^{k^2}$ by drawing the $k^2/10$ edges corresponding to the positions of 1s. For a pair of vertices $x \in X_{i,j}$ and $y \in Y_{i,j}$ that $(x,y)$ encodes $\ell_{i}$, Bob draws $(z_q,x)$ and $(z_q,y)$ for every $q \in [n]$. Note that the graph contains $2nbt + (k^2/10) \cdot bt = \Theta(m)$ edges.

\paragraph{Protocol $\pi$.} Alice starts the streaming algorithm $\mathcal{A}$ with an empty graph. She then inserts all of her edges and send the memory content to Bob. Bob then continues inserting all of his edges. In the end, Bob queries $\mathcal{A}$ for $\widetilde{T}$ such that $|T-\widetilde{T}| \leq \eps T$ with probability $0.99$ and answers $\frac{2\widetilde{T}}{bn}-t$.

\begin{claim} The protocol $\pi$ computes $\mathcal{P}$  with probability $0.99$.
\label{clm:pi_correct}
\end{claim}
\begin{proof} Consider the graph $G$ induced by blocks $(X_{i}, Y_{i})$ and  $z_1,\ldots,z_n$. If its corresponding instance $\mathsf{IND}(S_{i}, \ell_{i})$ answers $+1$, we shall discover $bn$ triangles as each of the $b$ mini-block produces exactly $n$ triangles. On the other hand, if the answer is $-1$, there is no triangles at all. Hence, we can write the number of triangles of $G$ as $T = \frac{bn}{2} \cdot \sum_{i \in [t]} \left(1+\mathsf{IND}_{k^2, k^2/10}(S_i, \ell_i)\right).$

The graph $G$ that Alice and Bob jointly constructs has $\Theta(n)$ vertices, $\Theta(m)$ edges, and the number of triangles in $G$ is $bn \cdot B(t,\frac{1}{2})$ which is $\Theta(bnt) = \Theta(\tau) = \Theta(T)$ with high constant probability. All of these satisfy the setting of $\tricount$. Thus, using $\mathcal{A}$, with probability 0.999, Bob answers $\widetilde{T}$ such that $|T - \widetilde{T}| \leq \eps T$. In other words, the additive error to $\mathcal{P}$ is $O\left(\frac{\eps T}{bn}\right) = O(1/\eps) \leq 0.01 \sqrt{t}$, using the fact that $\tau$ is a $O(1)$-approximation of $T$ and that we set $t = \Theta(1/\eps^2)$ for an appropriate leading constant. This means the protocol $\pi$ is correct with probability $0.99$, as we desired.
\end{proof}

\paragraph{Wrapping-up.} By chaining \Cref{clm:pi_lb} and \Cref{clm:pi_correct},  $\pi$ must use $\Omega\left(\frac{mn}{\eps^2T}\right)$ bits of communication. Since the sole message from Alice is the memory content of $\mathcal{A}$ which is $S$ bits, we have $S \geq \Omega\left(\frac{mn}{\eps^2T}\right)$ as desired.

\section{Acknowledgement}

The authors thank Elena Gribelyuk, Cheng Jiang, Yang Hu, Yinchen Liu, and Santhoshini Velusamy for helpful discussions that inspired this work. We are especially grateful to and credit Sepehr Assadi for pointing out the potential application of gap-majority lemmas to streaming triangle counting in \Cref{subsec:tricount}. We also thank anonymous FOCS 2026 reviewers for their valuable feedback and for identifying typographical errors in an earlier draft.
\vspace{3.5mm}

\noindent \textbf{AI Disclosure.} Numerical calculations in the proof of \Cref{lem:gapmaj_then_low_variance} (shown in Appendix \ref{appendix:BGW_reprove}) were assisted by ChatGPT.

% for discussions that influenced this work

\printbibliography

@inproceedings{SawettamalyaY25,
  author       = {Pachara Sawettamalya and
                  Huacheng Yu},
  title        = {Strong {XOR} Lemma for Information Complexity},
  booktitle    = {57th Annual {ACM} Symposium on Theory of Computing},
  pages        = {1626--1637},
  publisher    = {{ACM}},
  year         = {2025},
  url          = {https://doi.org/10.1145/3717823.3718258},
  doi          = {10.1145/3717823.3718258},
  bibsource    = {dblp computer science bibliography, https://dblp.org}
}

@inproceedings{Yu22,
  author       = {Huacheng Yu},
  title        = {Strong {XOR} Lemma for Communication with Bounded Rounds: (extended abstract)},
  booktitle    = {63rd {IEEE} Annual Symposium on Foundations of Computer Science},
  pages        = {1186--1192},
  publisher    = {{IEEE}},
  year         = {2022},
  url          = {https://doi.org/10.1109/FOCS54457.2022.00114},
  doi          = {10.1109/FOCS54457.2022.00114},
  bibsource    = {dblp computer science bibliography, https://dblp.org}
}

@inproceedings{BravermanGW20,
  author       = {Mark Braverman and
                  Sumegha Garg and
                  David P. Woodruff},
  title        = {The Coin Problem with Applications to Data Streams},
  booktitle    = {61st {IEEE} Annual Symposium on Foundations of Computer Science},
  pages        = {318--329},
  publisher    = {{IEEE}},
  year         = {2020},
  url          = {https://doi.org/10.1109/FOCS46700.2020.00038},
  doi          = {10.1109/FOCS46700.2020.00038},
  bibsource    = {dblp computer science bibliography, https://dblp.org}
}

@article{Braverman15,
  author       = {Mark Braverman},
  title        = {Interactive Information Complexity},
  journal      = {{SIAM} J. Comput.},
  volume       = {44},
  number       = {6},
  pages        = {1698--1739},
  year         = {2015},
  url          = {https://doi.org/10.1137/130938517},
  doi          = {10.1137/130938517},
  bibsource    = {dblp computer science bibliography, https://dblp.org}
}

@article{BYJKS04,
  author       = {Ziv Bar{-}Yossef and
                  T. S. Jayram and
                  Ravi Kumar and
                  D. Sivakumar},
  title        = {An information statistics approach to data stream and communication complexity},
  journal      = {J. Comput. Syst. Sci.},
  volume       = {68},
  number       = {4},
  pages        = {702--732},
  year         = {2004},
  url          = {https://doi.org/10.1016/j.jcss.2003.11.006},
  doi          = {10.1016/J.JCSS.2003.11.006},
  bibsource    = {dblp computer science bibliography, https://dblp.org}
}

@inproceedings{BBCR10,
  author       = {Boaz Barak and
                  Mark Braverman and
                  Xi Chen and
                  Anup Rao},
  title        = {How to compress interactive communication},
  booktitle    = {42nd {ACM} Symposium on Theory of Computing},
  pages        = {67--76},
  publisher    = {{ACM}},
  year         = {2010},
  url          = {https://doi.org/10.1145/1806689.1806701},
  doi          = {10.1145/1806689.1806701},
  bibsource    = {dblp computer science bibliography, https://dblp.org}
}

@inproceedings{IR24a,
  author       = {Siddharth Iyer and
                  Anup Rao},
  title        = {An {XOR} Lemma for Deterministic Communication Complexity},
  booktitle    = {65th {IEEE} Annual Symposium on Foundations of Computer Science},
  pages        = {429--432},
  publisher    = {{IEEE}},
  year         = {2024},
  url          = {https://doi.org/10.1109/FOCS61266.2024.00034},
  doi          = {10.1109/FOCS61266.2024.00034},
  bibsource    = {dblp computer science bibliography, https://dblp.org}
}

@inproceedings{IR24b,
  author       = {Siddharth Iyer and
                  Anup Rao},
  title        = {{XOR} Lemmas for Communication via Marginal Information},
  booktitle    = {56th Annual {ACM} Symposium on Theory of Computing},
  pages        = {652--658},
  publisher    = {{ACM}},
  year         = {2024},
  url          = {https://doi.org/10.1145/3618260.3649726},
  doi          = {10.1145/3618260.3649726},
  bibsource    = {dblp computer science bibliography, https://dblp.org}
}

@inproceedings{BRWY13,
  author       = {Mark Braverman and
                  Anup Rao and
                  Omri Weinstein and
                  Amir Yehudayoff},
  title        = {Direct Products in Communication Complexity},
  booktitle    = {54th Annual {IEEE} Symposium on Foundations of Computer Science},
  pages        = {746--755},
  publisher    = {{IEEE} Computer Society},
  year         = {2013},
  url          = {https://doi.org/10.1109/FOCS.2013.85},
  doi          = {10.1109/FOCS.2013.85},
  bibsource    = {dblp computer science bibliography, https://dblp.org}
}

@inproceedings{JPY12,
  author       = {Rahul Jain and
                  Attila Pereszl{\'{e}}nyi and
                  Penghui Yao},
  title        = {A Direct Product Theorem for the Two-Party Bounded-Round Public-Coin Communication Complexity},
  booktitle    = {53rd Annual {IEEE} Symposium on Foundations of Computer Science},
  pages        = {167--176},
  publisher    = {{IEEE} Computer Society},
  year         = {2012},
  url          = {https://doi.org/10.1109/FOCS.2012.42},
  doi          = {10.1109/FOCS.2012.42},
  bibsource    = {dblp computer science bibliography, https://dblp.org}
}

@inproceedings{BR11,
  author       = {Mark Braverman and
                  Anup Rao},
  title        = {Information Equals Amortized Communication},
  booktitle    = {{IEEE} 52nd Annual Symposium on Foundations of Computer Science},
  pages        = {748--757},
  publisher    = {{IEEE} Computer Society},
  year         = {2011},
  url          = {https://doi.org/10.1109/FOCS.2011.86},
  doi          = {10.1109/FOCS.2011.86},
  bibsource    = {dblp computer science bibliography, https://dblp.org}
}

@inproceedings{Yao82,
  author       = {Andrew Chi{-}Chih Yao},
  title        = {Theory and Applications of Trapdoor Functions (Extended Abstract)},
  booktitle    = {23rd Annual Symposium on Foundations of Computer Science},
  pages        = {80--91},
  publisher    = {{IEEE} Computer Society},
  year         = {1982},
  url          = {https://doi.org/10.1109/SFCS.1982.45},
  doi          = {10.1109/SFCS.1982.45},
  bibsource    = {dblp computer science bibliography, https://dblp.org}
}

@article{Lev87,
  author       = {Leonid A. Levin},
  title        = {One-way functions and pseudorandom generators},
  journal      = {Comb.},
  volume       = {7},
  number       = {4},
  pages        = {357--363},
  year         = {1987},
  url          = {https://doi.org/10.1007/BF02579323},
  doi          = {10.1007/BF02579323},
  bibsource    = {dblp computer science bibliography, https://dblp.org}
}

@inproceedings{IW97,
  author       = {Russell Impagliazzo and
                  Avi Wigderson},
  title        = {\emph{P = BPP} if \emph{E} Requires Exponential Circuits: Derandomizing the {XOR} Lemma},
  booktitle    = {29th {ACM} Symposium on the Theory
                  of Computing},
  pages        = {220--229},
  publisher    = {{ACM}},
  year         = {1997},
  url          = {https://doi.org/10.1145/258533.258590},
  doi          = {10.1145/258533.258590},
  bibsource    = {dblp computer science bibliography, https://dblp.org}
}

@inproceedings{Imp95,
  author       = {Russell Impagliazzo},
  title        = {Hard-Core Distributions for Somewhat Hard Problems},
  booktitle    = {36th Annual Symposium on Foundations of Computer Science},
  pages        = {538--545},
  publisher    = {{IEEE} Computer Society},
  year         = {1995},
  url          = {https://doi.org/10.1109/SFCS.1995.492584},
  doi          = {10.1109/SFCS.1995.492584},
  bibsource    = {dblp computer science bibliography, https://dblp.org}
}

@article{Dru12,
  author       = {Andrew Drucker},
  title        = {Improved direct product theorems for randomized query complexity},
  journal      = {Comput. Complex.},
  volume       = {21},
  number       = {2},
  pages        = {197--244},
  year         = {2012},
  url          = {https://doi.org/10.1007/s00037-012-0043-7},
  doi          = {10.1007/S00037-012-0043-7},
  bibsource    = {dblp computer science bibliography, https://dblp.org}
}

@article{BKLS20,
  author       = {Joshua Brody and
                  Jae Tak Kim and
                  Peem Lerdputtipongporn and
                  Hariharan Srinivasulu},
  title        = {A Strong {XOR} Lemma for Randomized Query Complexity},
  journal      = {CoRR},
  volume       = {abs/2007.05580},
  year         = {2020},
  url          = {https://arxiv.org/abs/2007.05580},
  eprinttype    = {arXiv},
  eprint       = {2007.05580},
  bibsource    = {dblp computer science bibliography, https://dblp.org}
}

@inproceedings{BesselmanGGM025,
  author       = {Tyler Besselman and
                  Mika G{\"{o}}{\"{o}}s and
                  Siyao Guo and
                  Gilbert Maystre and
                  Weiqiang Yuan},
  title        = {Direct Sums for Parity Decision Trees},
  booktitle    = {40th Computational Complexity Conference},
  series       = {LIPIcs},
  volume       = {339},
  pages        = {16:1--16:38},
  publisher    = {Schloss Dagstuhl - Leibniz-Zentrum f{\"{u}}r Informatik},
  year         = {2025},
  url          = {https://doi.org/10.4230/LIPIcs.CCC.2025.16},
  doi          = {10.4230/LIPICS.CCC.2025.16},
  bibsource    = {dblp computer science bibliography, https://dblp.org}
}

@article{Ben-DavidK18,
  author       = {Shalev Ben{-}David and
                  Robin Kothari},
  title        = {Randomized Query Complexity of Sabotaged and Composed Functions},
  journal      = {Theory Comput.},
  volume       = {14},
  number       = {1},
  pages        = {1--27},
  year         = {2018},
  url          = {https://doi.org/10.4086/toc.2018.v014a005},
  doi          = {10.4086/TOC.2018.V014A005},
  bibsource    = {dblp computer science bibliography, https://dblp.org}
}

@inproceedings{BlaisB19,
  author       = {Eric Blais and
                  Joshua Brody},
  title        = {Optimal Separation and Strong Direct Sum for Randomized Query Complexity},
  booktitle    = {34th Computational Complexity Conference},
  series       = {LIPIcs},
  volume       = {137},
  pages        = {29:1--29:17},
  publisher    = {Schloss Dagstuhl - Leibniz-Zentrum f{\"{u}}r Informatik},
  year         = {2019},
  url          = {https://doi.org/10.4230/LIPIcs.CCC.2019.29},
  doi          = {10.4230/LIPICS.CCC.2019.29},
  bibsource    = {dblp computer science bibliography, https://dblp.org}
}

@article{JainKS10,
  author       = {Rahul Jain and
                  Hartmut Klauck and
                  Miklos Santha},
  title        = {Optimal direct sum results for deterministic and randomized decision
                  tree complexity},
  journal      = {Inf. Process. Lett.},
  volume       = {110},
  number       = {20},
  pages        = {893--897},
  year         = {2010},
  url          = {https://doi.org/10.1016/j.ipl.2010.07.020},
  doi          = {10.1016/J.IPL.2010.07.020},
  bibsource    = {dblp computer science bibliography, https://dblp.org}
}

@article{Savick02,
  author       = {Petr Savick{\'{y}}},
  title        = {On determinism versus unambiquous nondeterminism for decision trees},
  journal      = {Electron. Colloquium Comput. Complex.},
  volume       = {{TR02-009}},
  year         = {2002},
  url          = {https://eccc.weizmann.ac.il/eccc-reports/2002/TR02-009/index.html},
  eprinttype    = {ECCC},
  eprint       = {TR02-009},
  bibsource    = {dblp computer science bibliography, https://dblp.org}
}

@inproceedings{LeeMRSS11,
  author       = {Troy Lee and
                  Rajat Mittal and
                  Ben W. Reichardt and
                  Robert Spalek and
                  Mario Szegedy},
  title        = {Quantum Query Complexity of State Conversion},
  booktitle    = {{IEEE} 52nd Annual Symposium on Foundations of Computer Science},
  pages        = {344--353},
  publisher    = {{IEEE} Computer Society},
  year         = {2011},
  url          = {https://doi.org/10.1109/FOCS.2011.75},
  doi          = {10.1109/FOCS.2011.75},
  bibsource    = {dblp computer science bibliography, https://dblp.org}
}

@inproceedings{Reichardt11a,
  author       = {Ben Reichardt},
  title        = {Reflections for quantum query algorithms},
  booktitle    = {22nd Annual {ACM-SIAM} Symposium on Discrete
                  Algorithms},
  pages        = {560--569},
  publisher    = {{SIAM}},
  year         = {2011},
  url          = {https://doi.org/10.1137/1.9781611973082.44},
  doi          = {10.1137/1.9781611973082.44},
  bibsource    = {dblp computer science bibliography, https://dblp.org}
}

@inproceedings{Ben-DavidGK020,
  author       = {Shalev Ben{-}David and
                  Mika G{\"{o}}{\"{o}}s and
                  Robin Kothari and
                  Thomas Watson},
  title        = {When Is Amplification Necessary for Composition in Randomized Query Complexity?},
  booktitle    = {Approximation, Randomization, and Combinatorial Optimization},
  series       = {LIPIcs},
  volume       = {176},
  pages        = {28:1--28:16},
  publisher    = {Schloss Dagstuhl - Leibniz-Zentrum f{\"{u}}r Informatik},
  year         = {2020},
  url          = {https://doi.org/10.4230/LIPIcs.APPROX/RANDOM.2020.28},
  doi          = {10.4230/LIPICS.APPROX/RANDOM.2020.28},
  bibsource    = {dblp computer science bibliography, https://dblp.org}
}

@article{FeigeRPU94,
  author       = {Uriel Feige and
                  Prabhakar Raghavan and
                  David Peleg and
                  Eli Upfal},
  title        = {Computing with Noisy Information},
  journal      = {{SIAM} J. Comput.},
  volume       = {23},
  number       = {5},
  pages        = {1001--1018},
  year         = {1994},
  url          = {https://doi.org/10.1137/S0097539791195877},
  doi          = {10.1137/S0097539791195877},
  bibsource    = {dblp computer science bibliography, https://dblp.org}
}

@inproceedings{GoosM21,
  author       = {Mika G{\"{o}}{\"{o}}s and
                  Gilbert Maystre},
  title        = {A Majority Lemma for Randomised Query Complexity},
  booktitle    = {36th Computational Complexity Conference},
  series       = {LIPIcs},
  volume       = {200},
  pages        = {18:1--18:15},
  publisher    = {Schloss Dagstuhl - Leibniz-Zentrum f{\"{u}}r Informatik},
  year         = {2021},
  url          = {https://doi.org/10.4230/LIPIcs.CCC.2021.18},
  doi          = {10.4230/LIPICS.CCC.2021.18},
  bibsource    = {dblp computer science bibliography, https://dblp.org}
}

@inproceedings{MackenzieS25,
  author       = {Simon Mackenzie and
                  Abdallah Saffidine},
  title        = {Refuting the Direct Sum Conjecture for Total Functions in Deterministic Communication Complexity},
  booktitle    = {57th Annual {ACM} Symposium on Theory of Computing},
  pages        = {572--583},
  publisher    = {{ACM}},
  year         = {2025},
  url          = {https://doi.org/10.1145/3717823.3718148},
  doi          = {10.1145/3717823.3718148},
  bibsource    = {dblp computer science bibliography, https://dblp.org}
}

@article{ChakrabartiR12,
  author       = {Amit Chakrabarti and
                  Oded Regev},
  title        = {An Optimal Lower Bound on the Communication Complexity of Gap-Hamming-Distance},
  journal      = {{SIAM} J. Comput.},
  volume       = {41},
  number       = {5},
  pages        = {1299--1317},
  year         = {2012},
  url          = {https://doi.org/10.1137/120861072},
  doi          = {10.1137/120861072},
  bibsource    = {dblp computer science bibliography, https://dblp.org}
}

@article{Sherstov12,
  author       = {Alexander A. Sherstov},
  title        = {The Communication Complexity of Gap Hamming Distance},
  journal      = {Theory Comput.},
  volume       = {8},
  number       = {1},
  pages        = {197--208},
  year         = {2012},
  url          = {https://doi.org/10.4086/toc.2012.v008a008},
  doi          = {10.4086/TOC.2012.V008A008},
  bibsource    = {dblp computer science bibliography, https://dblp.org}
}

@article{Vidick12,
  author       = {Thomas Vidick},
  title        = {A concentration inequality for the overlap of a vector on a large set, with application to the communication complexity of the Gap-Hamming-Distance problem},
  journal      = {Chic. J. Theor. Comput. Sci.},
  volume       = {2012},
  year         = {2012},
  url          = {http://cjtcs.cs.uchicago.edu/articles/2012/1/contents.html},
  bibsource    = {dblp computer science bibliography, https://dblp.org}
}

@inproceedings{ChakrabartiKW12,
  author       = {Amit Chakrabarti and
                  Ranganath Kondapally and
                  Zhenghui Wang},
  title        = {Information Complexity versus Corruption and Applications to Orthogonality and Gap-Hamming},
  booktitle    = {Approximation, Randomization, and Combinatorial Optimization},
  series       = {Lecture Notes in Computer Science},
  volume       = {7408},
  pages        = {483--494},
  publisher    = {Springer},
  year         = {2012},
  url          = {https://doi.org/10.1007/978-3-642-32512-0\_41},
  doi          = {10.1007/978-3-642-32512-0\_41},
  bibsource    = {dblp computer science bibliography, https://dblp.org}
}

@inproceedings{IndykW05,
  author       = {Piotr Indyk and
                  David P. Woodruff},
  title        = {Optimal approximations of the frequency moments of data streams},
  booktitle    = {37th Annual {ACM} Symposium on Theory of Computing},
  pages        = {202--208},
  publisher    = {{ACM}},
  year         = {2005},
  url          = {https://doi.org/10.1145/1060590.1060621},
  doi          = {10.1145/1060590.1060621},
  bibsource    = {dblp computer science bibliography, https://dblp.org}
}

@inproceedings{JowhariG05,
  author       = {Hossein Jowhari and
                  Mohammad Ghodsi},
  title        = {New Streaming Algorithms for Counting Triangles in Graphs},
  booktitle    = {11th Annual International Computing and Combinatorics Conference},
  series       = {Lecture Notes in Computer Science},
  volume       = {3595},
  pages        = {710--716},
  publisher    = {Springer},
  year         = {2005},
  url          = {https://doi.org/10.1007/11533719\_72},
  doi          = {10.1007/11533719\_72},
  bibsource    = {dblp computer science bibliography, https://dblp.org}
}

@inproceedings{Bar-YossefKS02,
  author       = {Ziv Bar{-}Yossef and
                  Ravi Kumar and
                  D. Sivakumar},
  title        = {Reductions in streaming algorithms, with an application to counting triangles in graphs},
  booktitle    = {30th {ACM-SIAM} Symposium on Discrete
                  Algorithms},
  pages        = {623--632},
  publisher    = {{ACM/SIAM}},
  year         = {2002},
  url          = {http://dl.acm.org/citation.cfm?id=545381.545464},
  bibsource    = {dblp computer science bibliography, https://dblp.org}
}

@inproceedings{BuriolFLMS06,
  author       = {Luciana S. Buriol and
                  Gereon Frahling and
                  Stefano Leonardi and
                  Alberto Marchetti{-}Spaccamela and
                  Christian Sohler},
  title        = {Counting triangles in data streams},
  booktitle    = {25th {ACM} {SIGACT-SIGMOD-SIGART} Symposium
                  on Principles of Database Systems},
  pages        = {253--262},
  publisher    = {{ACM}},
  year         = {2006},
  url          = {https://doi.org/10.1145/1142351.1142388},
  doi          = {10.1145/1142351.1142388},
  bibsource    = {dblp computer science bibliography, https://dblp.org}
}

@inproceedings{BravermanOV13,
  author       = {Vladimir Braverman and
                  Rafail Ostrovsky and
                  Dan Vilenchik},
  title        = {How Hard Is Counting Triangles in the Streaming Model?},
  booktitle    = {40th International Colloquium on Automata, Languages, and Programming},
  series       = {Lecture Notes in Computer Science},
  volume       = {7965},
  pages        = {244--254},
  publisher    = {Springer},
  year         = {2013},
  url          = {https://doi.org/10.1007/978-3-642-39206-1\_21},
  doi          = {10.1007/978-3-642-39206-1\_21},
  bibsource    = {dblp computer science bibliography, https://dblp.org}
}

@inproceedings{KallaugherP17,
  author       = {John Kallaugher and
                  Eric Price},
  title        = {A Hybrid Sampling Scheme for Triangle Counting},
  booktitle    = {28th Annual {ACM-SIAM} Symposium on Discrete
                  Algorithms},
  pages        = {1778--1797},
  publisher    = {{SIAM}},
  year         = {2017},
  url          = {https://doi.org/10.1137/1.9781611974782.116},
  doi          = {10.1137/1.9781611974782.116},
  bibsource    = {dblp computer science bibliography, https://dblp.org}
}

@inproceedings{JayaramK21,
  author       = {Rajesh Jayaram and
                  John Kallaugher},
  title        = {An Optimal Algorithm for Triangle Counting in the Stream},
  booktitle    = {Approximation, Randomization, and Combinatorial Optimization},
  series       = {LIPIcs},
  volume       = {207},
  pages        = {11:1--11:11},
  publisher    = {Schloss Dagstuhl - Leibniz-Zentrum f{\"{u}}r Informatik},
  year         = {2021},
  url          = {https://doi.org/10.4230/LIPIcs.APPROX/RANDOM.2021.11},
  doi          = {10.4230/LIPICS.APPROX/RANDOM.2021.11},
  bibsource    = {dblp computer science bibliography, https://dblp.org}
}

@article{JayramW13,
  author       = {T. S. Jayram and
                  David P. Woodruff},
  title        = {Optimal Bounds for Johnson-Lindenstrauss Transforms and Streaming Problems with Subconstant Error},
  journal      = {{ACM} Trans. Algorithms},
  volume       = {9},
  number       = {3},
  pages        = {26:1--26:17},
  year         = {2013},
  url          = {https://doi.org/10.1145/2483699.2483706},
  doi          = {10.1145/2483699.2483706},
  bibsource    = {dblp computer science bibliography, https://dblp.org}
}

@inproceedings{Woodruff04,
  author       = {David P. Woodruff},
  title        = {Optimal space lower bounds for all frequency moments},
  booktitle    = {15th Annual {ACM-SIAM} Symposium on Discrete
                  Algorithms},
  pages        = {167--175},
  publisher    = {{SIAM}},
  year         = {2004},
  url          = {http://dl.acm.org/citation.cfm?id=982792.982817},
  bibsource    = {dblp computer science bibliography, https://dblp.org}
}

@article{AlonMS99,
  author       = {Noga Alon and
                  Yossi Matias and
                  Mario Szegedy},
  title        = {The Space Complexity of Approximating the Frequency Moments},
  journal      = {J. Comput. Syst. Sci.},
  volume       = {58},
  number       = {1},
  pages        = {137--147},
  year         = {1999},
  url          = {https://doi.org/10.1006/jcss.1997.1545},
  doi          = {10.1006/JCSS.1997.1545},
  bibsource    = {dblp computer science bibliography, https://dblp.org}
}

@article{BravermanGPW16,
  author       = {Mark Braverman and
                  Ankit Garg and
                  Denis Pankratov and
                  Omri Weinstein},
  title        = {Information Lower Bounds via Self-Reducibility},
  journal      = {Theory Comput. Syst.},
  volume       = {59},
  number       = {2},
  pages        = {377--396},
  year         = {2016},
  url          = {https://doi.org/10.1007/s00224-015-9655-z},
  doi          = {10.1007/S00224-015-9655-Z},
  bibsource    = {dblp computer science bibliography, https://dblp.org}
}

@article{KerenidisLLRX15,
  author       = {Iordanis Kerenidis and
                  Sophie Laplante and
                  Virginie Lerays and
                  J{\'{e}}r{\'{e}}mie Roland and
                  David Xiao},
  title        = {Lower Bounds on Information Complexity via Zero-Communication Protocols
                  and Applications},
  journal      = {{SIAM} J. Comput.},
  volume       = {44},
  number       = {5},
  pages        = {1550--1572},
  year         = {2015},
  url          = {https://doi.org/10.1137/130928273},
  doi          = {10.1137/130928273},
  bibsource    = {dblp computer science bibliography, https://dblp.org}
}

@book{Rao_Yehudayoff_2020, place={Cambridge}, title={Communication Complexity: and Applications}, publisher={Cambridge University Press}, author={Rao, Anup and Yehudayoff, Amir}, year={2020}}

@inproceedings{AndoniCKQWZ16,
  author       = {Alexandr Andoni and
                  Jiecao Chen and
                  Robert Krauthgamer and
                  Bo Qin and
                  David P. Woodruff and
                  Qin Zhang},
  title        = {On Sketching Quadratic Forms},
  booktitle    = {Innovations in Theoretical Computer Science},
  pages        = {311--319},
  publisher    = {{ACM}},
  year         = {2016},
  url          = {https://doi.org/10.1145/2840728.2840753},
  doi          = {10.1145/2840728.2840753},
  bibsource    = {dblp computer science bibliography, https://dblp.org}
}

@article{Razborov92,
  author       = {Alexander A. Razborov},
  title        = {On the Distributional Complexity of Disjointness},
  journal      = {Theor. Comput. Sci.},
  volume       = {106},
  number       = {2},
  pages        = {385--390},
  year         = {1992},
  url          = {https://doi.org/10.1016/0304-3975(92)90260-M},
  doi          = {10.1016/0304-3975(92)90260-M},
  bibsource    = {dblp computer science bibliography, https://dblp.org}
}

@article{KalyanasundaramS92,
  author       = {Bala Kalyanasundaram and
                  Georg Schnitger},
  title        = {The Probabilistic Communication Complexity of Set Intersection},
  journal      = {{SIAM} J. Discret. Math.},
  volume       = {5},
  number       = {4},
  pages        = {545--557},
  year         = {1992},
  url          = {https://doi.org/10.1137/0405044},
  doi          = {10.1137/0405044},
  bibsource    = {dblp computer science bibliography, https://dblp.org}
}

@inproceedings{IndykW03,
  author       = {Piotr Indyk and
                  David P. Woodruff},
  title        = {Tight Lower Bounds for the Distinct Elements Problem},
  booktitle    = {44th Symposium on Foundations of Computer Science},
  pages        = {283--288},
  publisher    = {{IEEE} Computer Society},
  year         = {2003},
  url          = {https://doi.org/10.1109/SFCS.2003.1238202},
  doi          = {10.1109/SFCS.2003.1238202},
  bibsource    = {dblp computer science bibliography, https://dblp.org}
}

@article{BlaisBM12,
  author       = {Eric Blais and
                  Joshua Brody and
                  Kevin Matulef},
  title        = {Property Testing Lower Bounds via Communication Complexity},
  journal      = {Comput. Complex.},
  volume       = {21},
  number       = {2},
  pages        = {311--358},
  year         = {2012},
  url          = {https://doi.org/10.1007/s00037-012-0040-x},
  doi          = {10.1007/S00037-012-0040-X},
  bibsource    = {dblp computer science bibliography, https://dblp.org}
}

@article{FederKNN95,
  author       = {Tom{\'{a}}s Feder and
                  Eyal Kushilevitz and
                  Moni Naor and
                  Noam Nisan},
  title        = {Amortized Communication Complexity},
  journal      = {{SIAM} J. Comput.},
  volume       = {24},
  number       = {4},
  pages        = {736--750},
  year         = {1995},
  url          = {https://doi.org/10.1137/S0097539792235864},
  doi          = {10.1137/S0097539792235864},
  bibsource    = {dblp computer science bibliography, https://dblp.org}
}

@inproceedings{Tal13,
  author       = {Avishay Tal},
  title        = {Properties and applications of boolean function composition},
  booktitle    = {Innovations in Theoretical Computer Science},
  pages        = {441--454},
  publisher    = {{ACM}},
  year         = {2013},
  url          = {https://doi.org/10.1145/2422436.2422485},
  doi          = {10.1145/2422436.2422485},
  bibsource    = {dblp computer science bibliography, https://dblp.org}
}

@article{Montanaro14,
  author       = {Ashley Montanaro},
  title        = {A composition theorem for decision tree complexity},
  journal      = {Chic. J. Theor. Comput. Sci.},
  volume       = {2014},
  year         = {2014},
  url          = {http://cjtcs.cs.uchicago.edu/articles/2014/6/contents.html},
  bibsource    = {dblp computer science bibliography, https://dblp.org}
}

@inproceedings{BenDavidB25,
  author       = {Shalev Ben{-}David and
                  Eric Blais},
  title        = {Direct Product Theorems for Randomized Query Complexity},
  booktitle    = {66th {IEEE} Annual Symposium on Foundations of Computer Science},
  pages        = {710--733},
  publisher    = {{IEEE}},
  year         = {2025},
  url          = {https://doi.org/10.1109/FOCS63196.2025.00038},
  doi          = {10.1109/FOCS63196.2025.00038},
  bibsource    = {dblp computer science bibliography, https://dblp.org}
}

\appendix 

\section{Proof of \Cref{lem:gapmaj_then_low_variance}}
\label{appendix:BGW_reprove}

We give the missing proof of \Cref{lem:gapmaj_then_low_variance}, which closely follows Claim 6 of \cite{BravermanGW20}. 

% Note that some of the numerical calculations are assisted by ChatGPT.

\variancebound*

\begin{proof} Denote $\pi = (X,Y,M,R)$ where $X = (X_1,\ldots,X_n)$ and $Y = (Y_1,\ldots,Y_n)$ such that each $(X_i,Y_i) \sim \mu$ i.i.d. We may assume that the final message of $M$ is $A \in \{\pm 1\}$ indicating the players' answer to $\GapMAJ_n \circ f^n(X,Y)$. We start by expanding the observer's variance of $\pi$ on $f^{+n}$.  
\begin{align}
    \VarObs(\pi @ f^{+n}) & = \Esub{MR} \ \Varsub{XY}\left[f^+(X,Y) \mid MR\right] \nonumber \\ 
    & \leq \Esub{A} \ \Varsub{XY}\left[f^+(X,Y) \mid A\right] \tag{\Cref{fact:cond_var} and $M$ contains $A$} \nonumber \\
    & = \Esub{XY}\left[f^+(X,Y)^2\right] - \Esub{A} \left[\Esub{XY}\left[f^+(X,Y) \mid A\right]^2\right] \nonumber \\
    & = n - \Esub{A} \left[\Esub{XY}\left[f^+(X,Y) \mid A\right]^2\right] \tag{each $f(X_i,Y_i)$ is uniform $\{\pm 1\}$ i.i.d.}\nonumber \\
    & = n - \Esub{A} \left[\Esub{Z}\left[\langle Z,\mathbf{1}_n\rangle \mid A\right]^2\right] \label{eq:boundvar1}
\end{align}
where we denote $Z = (Z_1,\ldots,Z_n)$ such that each $Z_i = f(X_i,Y_i)$ is i.i.d. equally likely to take values $\{\pm 1\}$ under $\mu$. This also mean for any $z \in \{\pm 1\}^n$, we have $\pi(z) = 2^{-n}$. Also denote $p_z := \pi(A=1 \mid z)$. As $\pi$ correct with probability 0.99, we have:
\begin{align*}
    0.01 & \geq 2^{-n} \cdot \left(\sum_{\langle z, \mathbf{1}_n\rangle \geq 0.01\sqrt{n}} (1-p_z) + \sum_{\langle z, \mathbf{1}_n\rangle  \leq -0.01\sqrt{n}}p_z\right)
\end{align*} 
which implies $\left(\sum_{\langle z, \mathbf{1}_n\rangle \geq 0.01\sqrt{n}} p_z - \sum_{\langle z, \mathbf{1}_n\rangle \leq -0.01\sqrt{n}}p_z\right) \geq 0.486 \cdot 2^n$ using the fact that $|\{z \mid \langle z, \mathbf{1}_n\rangle \geq 0.01\sqrt{n}\} \geq 0.496 \cdot 2^n$. This further gives:
\begin{equation}
    \sum_{\langle z, \mathbf{1}_n\rangle \geq 0.01\sqrt{n}}p_z \geq 0.486 \cdot 2^n \hspace{3mm} \text{ and } \hspace{3mm}   \sum_{\langle z, \mathbf{1}_n\rangle \leq -0.01\sqrt{n}}p_z \leq 0.014 \cdot 2^n \label{eq:bound_pzs}
\end{equation}
% \begin{itemize}
%     \item $\sum_{z \geq 0.01\sqrt{n}}p_z \geq 2^n \cdot (0.498 - \frac{\eps}{2})$
%     \item $\sum_{z \leq -0.01\sqrt{n}} p_z \leq 2^n \cdot \frac{1-\eps}{2}- 2^n \cdot (0.498 - \frac{\eps}{2}) = 2^n \cdot 0.02$
% \end{itemize}

Our next step is to show $\Esub{z}\left[\langle z, \mathbf{1}_n\rangle \mid A=1\right] \geq \Omega(\sqrt{n})$. We begin by expanding the expectation.
\begin{align}
    \Esub{Z}\left[\langle Z, \mathbf{1}_n\rangle \mid A=1\right] & = \frac{\sum_{z} p_z \langle z, \mathbf{1}_n\rangle}{ \sum_{z} p_z} \geq 2^{-n} \cdot \sum_{z} p_z \langle z, \mathbf{1}_n\rangle \label{eq:exp_A1}
\end{align}
where the final inequality follows $\sum_{z} p_z \cdot 2^{-n} = \pi(A = 1) \leq 1$. We then expand
\begin{align*}
    \sum_{z} p_z \langle z,\mathbf{1}_n\rangle & = \underbrace{\sum_{\langle z,\mathbf{1}_n\rangle \geq 0.01 \sqrt{n}} p_z \langle z,\mathbf{1}_n\rangle}_{\Phi_1} + \underbrace{\sum_{\langle z,\mathbf{1}_n\rangle \leq -0.01 \sqrt{n}} p_z \langle z,\mathbf{1}_n\rangle }_{\Phi_2}
    \\ &  \hspace{1cm}+ \underbrace{\sum_{0 \leq \langle z,\mathbf{1}_n\rangle < 0.01 \sqrt{n}} p_z \langle z,\mathbf{1}_n\rangle}_{\Phi_3} + \underbrace{\sum_{-0.01 \sqrt{n} \leq \langle z,\mathbf{1}_n\rangle \leq 0} p_z \langle z,\mathbf{1}_n\rangle}_{\Phi_4}.
\end{align*}
Notice that $\Phi_3 \geq 0$. Utilizing \Cref{eq:bound_pzs} along with the facts that $\sum_{i = 0.005\sqrt{n}}^{1.162\sqrt{n}} \binom{n}{n/2+i} \leq 0.486 \cdot 2^n$ and $\sum_{i = 0}^{n/2-1.098\sqrt{n}} \binom{n}{i} \geq 0.014 \cdot 2^n$, we obtain the following bounds:
$$\Phi_4 \geq - \sum_{0 \leq \langle z,\mathbf{1}_n\rangle < 0.01 \sqrt{n}} \langle z,\mathbf{1}_n\rangle = - \sum_{i=0}^{0.005 \sqrt{n}} \binom{n}{n/2+i} \cdot 2i$$
$$\Phi_1 \geq \sum_{i = 0.005 \sqrt{n}}^{1.162\sqrt{n}} \binom{n}{n/2+i} \cdot 2i$$
$$\Phi_2 \geq - \sum_{i = 0}^{n/2-1.098\sqrt{n}} \binom{n}{i} \cdot (n-2i) = \sum_{i = 1.098\sqrt{n}}^{n/2} \binom{n}{n/2+i} \cdot 2i $$

Combining these terms, we have:
\begin{align*}
    \sum_{z} p_z \langle z,\mathbf{1}_n\rangle  & \geq \sum_{i = 0.005 \sqrt{n}}^{1.162\sqrt{n}} \binom{n}{n/2+i} \cdot 2i - \sum_{i =  1.098\sqrt{n}}^{n/2} \binom{n}{n/2+i} \cdot 2i - \sum_{i=0}^{0.005 \sqrt{n}} \binom{n}{n/2+i} \cdot 2i \\
    & = 2^n \cdot \sqrt{\frac{n}{2\pi}} \cdot \left[\phi(0.005) -\phi(1.162) - \phi(1.098) - \phi(0) + \phi(0.005)\right] - o(2^n \sqrt{n}) \tag{\Cref{fact:tailgaussian}}\\
    & \geq 0.6724 \cdot 2^n \cdot \sqrt{n}
\end{align*}

Plugging this back into \Cref{eq:exp_A1}, we have $\Esub{Z}\left[\langle Z, \mathbf{1}_n\rangle \mid A=1\right] \geq 0.6724\sqrt{n}.$ A symmetric argument  gives $\Esub{Z}\left[\langle Z, \mathbf{1}_n\rangle \mid A=-1\right] \geq 0.6724\sqrt{n}.$ These imply $\Esub{A} \ \Esub{Z}\left[\langle Z,\mathbf{1}_n\rangle \mid A \right] \geq 0.6724\sqrt{n}$. Combining with \Cref{eq:boundvar1}, we achieve $\VarObs(\pi @ f^{+n}) \leq n - (0.6724\sqrt{n})^2 \leq 0.99n$.
\end{proof}

\begin{fact} For any constant $c > 0$, we have $\sum_{i = 0}^{c \sqrt{n}} \binom{n}{n/2+i} \cdot i = 2^n \cdot \sqrt{\frac{n}{8 \pi}} \cdot (1-\phi(c)) + o(2^n\sqrt{n})$ where $\phi(c) := e^{-2c^2}.$
\label{fact:tailgaussian}
\end{fact}

\end{document}